\documentclass[11pt,oneside]{article}
\usepackage{thmtools}
\usepackage{xparse}

\usepackage[margin=1in]{geometry}
\usepackage{etoolbox}

\providebool{bookdraft}%

\NewDocumentEnvironment{When}{m +b}{%
  \providebool{#1}%
  \ifbool{#1}{#2}{}%
}{}%

\NewDocumentEnvironment{Unless}{m +b}{%
  \providebool{#1}%
  \ifbool{#1}{}{#2}%
}{}%

\usepackage[hyphens]{url}

\usepackage[style=alphabetic,natbib=true,maxalphanames=10,maxnames=10,sortcites=true,doi=false,url=false]{biblatex} %

\usepackage{amsmath}
\usepackage[colorlinks,allcolors=red]{hyperref}%
\usepackage{amsthm}
\usepackage{cleveref}

\crefname{invariantsi}{invariant}{invariants}

\usepackage{fancyhdr}

\usepackage{lmodern}
\usepackage[full]{textcomp}     %

\usepackage{microtype}          %
\usepackage[T1]{fontenc}

\usepackage{article} %
\usepackage{math} %
\usepackage{algo} %
\usepackage{wrapfig} %
\usepackage{placeins} %
\usepackage{advdate}%

\usepackage{appendix}

\usepackage{graphicx} %
\usepackage{layout} %
\usepackage{setspace} %
\usepackage{mathtools} %
\usepackage{grffile} %
\usepackage{pdfpages} %
\usepackage{mdframed} %
\usepackage{bm} %
\usepackage{needspace} %
\usepackage{pbox} %

\usepackage{titlesec}

\usepackage{truncate}

\NewDocumentEnvironment{Note}{+b}{%
  \begin{When}{bookdraft}%
    \footnote{#1}
  \end{When}
}{}%

\NewDocumentCommand{\undernote}{s O{blue} m m}{
  \IfBooleanT{#1}{\smash}%
  {\color{#2} %
    \underbrace{\normalcolor%
      #4}_{\mathclap{\text{#3}}} %
  }%
  \IfBooleanT{#1}{\vphantom{#4}}
}

\newcommand{\defterm}[1]{{\boldmath\normalfont \bfseries #1}}%
\renewcommand{\defterm}{\emph}%

\makeatletter
\g@addto@macro\bfseries{\boldmath}
\makeatother

\NewDocumentEnvironment{rightfigmp}{O{1em} m m}{
  \noindent%
  \begin{minipage}[t]{\linewidth}%
    \noindent%
    \begin{minipage}[t]{\linewidth - #2 - #1}%
    }{
    \end{minipage}%
    \hfill%
    \begin{minipage}[t]{#2}%
      \includegraphics[width=#2,valign=t]{figures/#3}%
    \end{minipage}
  \end{minipage}
}

\titlespacing*{\paragraph}{%
  0pt}{
  {\medskipamount}}{
  1em}

\providecommand{\wrt}{with respect to\xspace}%

\bibliography{references}

\NewDocumentEnvironment{Inset}{}{
  \fontsize{10pt}{12pt}\selectfont
  \renewcommand{\paragraph}[1]{

    \medskip\noindent\textbf{##1}}%
  \providecommand{\subparagraph}[1]{

    \medskip\noindent\textit{##1}}%
  \renewcommand{\subparagraph}[1]{

    \medskip\noindent\textit{##1}}%

  \noindent%
  \begin{oframed}%
  }{
  \end{oframed}%
}

\newcommand{\chapter}[1]{}

\definecolor{calypso}{RGB}{50, 104, 145}

\hypersetup{allcolors=calypso}

\definecolor{almostblack}{RGB}{18, 18, 18} %

\color{almostblack}%

\providebool{StandalonePaper}
\booltrue{StandalonePaper}

\bibliography{references}

\input{cuts.sty}
\input{matroids.sty}
\input{matroid-push-relabel.sty}

\newcommand{\FM}{\cite{frank-miklos}\xspace}%

\providecommand{\poly}{\fparnew{\operatorname{poly}}}

\newcommand{\MatroidUnionTime}{\bigO{n + \opt r \log{kr}}}%

\newcommand{\ApxMatroidUnionTime}{\bigO{n + \opt \log{k \therank} / \eps}}

\newcommand{\ApxBasePackingTime}{\bigO{n + k \therank \log{n} \log{k
    \therank} / \eps}}%

\newcommand{\RandomizedApxBasePackingTime}{\bigO{n +  r \log{n}^2
    \parof{\log \log{n} + \log{1/\eps}} /
    \eps^3}}%

\newcommand{\ApxBaseCoveringTime}{\bigO{n \log{n} \log{kr} /
    \eps}}%
\newcommand{\RandomizedApxCoveringTime}{\bigO{n + \therank \ln{n} \ln{\therank / \eps} / \eps^3}}

\newcommand{\CunninghamMatroidUnionTime}{\bigO{\parof{\opt^{3/2} + k}
    n Q + \opt^{1/2} k n}}

\newcommand{\CunninghamOptimalReinforcementTime}{\bigO{\parof{(kr)^{3/2}
      + k} n^2 Q + k^{3/2} r n^2}}

\newcommand{\RandomizedApxMatroidUnionTime}{\bigO{n + \parof{\opt /
      k} \log{n} \log{\therank \log{n} / \eps^2} /
    \eps^3}}

\newcommand{\KargerBasePackingTime}{\bigO{n + \therank^{3/2}
    \numM^{5/2} \log{k\therank}^{3/2}}} %

\newcommand{\KargerDeterministicBasePackingTime}{\bigO{n +
    \parof{\therank k}^3 \log{kr}^{\bigO{1}}}} %

\newcommand{\ForestUnionTime}{\bigO{m \ack{n} + n \opt \parof{\log{n}
      + \log{k} \ack{n}}}} %
\newcommand{\ForestUnionTimeB}{\bigO{m \ack{n} + \opt^{3/2} \log{n k}^2 \ack{n}^2}} %

\newcommand{\ApxForestUnionTime}{\bigO{m \ack{n} + \opt \log{n}
    \parof{\log{n} + \log{k} \ack{n}} / \eps}}

\newcommand{\RandomizedApxForestUnionTime}{\bigO{m \ack{n} + n
    \log{n}^2 \parof{\log{n} + \log{k} \ack{n}} / \eps^3}}

\newcommand{\ApxTreePackingTime}{\bigO{n + \parof{\opt / k} \log{n}
    \log{\therank \log{n} / \eps^2}}} %
\newcommand{\RandomizedApxTreePackingTime}{\bigO{n + \parof{\opt / k}
    \log{n} \log{\therank \log{n} / \eps^2}}}
\newcommand{\ApxTreeCoveringTime}{\bigO{m \ack{n} + n k \log{n}
    \parof{\log{n} + \log{k} \ack{n}} / \eps}} %
\newcommand{\RandomizedApxTreeCoveringTime}{ \bigO{m \ack{n} + n
    \ln{n} \parof{\log{n} + \log{\log{n} / \eps} \ack{n} } / \eps^3}
} %

\newcommand{\RandomizedApxMatroidStrengthTime}{%
  \bigO{\parof{n + r \log{n} \log{\therank}} \log{n U / \therank} +
    \therank \ln{n} \ln{\therank / \eps} / \eps^3 } %
}                                                   %

\newcommand{\MatroidUnionWithAugmentingPathsTime}{\bigO{n + \opt
    \sqrt{n' \log{kr}}}}

\newcommand{\MatroidUnionWithAugmentingPathsNotation}{n' = \min{n +
    \opt \log{\therank}, \therank k \log{k \therank}}}

\begin{document}

\title{Faster exact and approximation algorithms for
  packing~and~covering~matroids via push-relabel\footnote{Manuscript
    submitted to a conference in July 2022. Small edits made in a
    rebuttal phase in September 2022.}}

\author{Kent
  Quanrud\footnote{Purdue University. \algo{krq@purdue.edu}. Supported
    in part by NSF grant CCF-2129816.}}

\date{}

\maketitle

\begin{abstract}
  Matroids are a fundamental object of study in combinatorial
  optimization. Three closely related and important problems involving
  matroids are maximizing the size of the union of $k$ independent
  sets (that is, \emph{$k$-fold matroid union}), computing $k$
  disjoint bases (a.k.a.\ \emph{matroid base packing}), and covering
  the elements by $k$ bases (a.k.a.\ \emph{matroid base
    covering}). These problems generalize naturally to integral and
  real-valued capacities on the elements.  This work develops faster
  exact and/or approximation problems for these and some other closely
  related problems such as optimal reinforcement and matroid
  membership. We obtain improved running times both for general
  matroids in the independence oracle model and for the graphic
  matroid.  The main thrust of our improvements comes from developing
  a faster and unifying \emph{push-relabel} algorithm for the
  integer-capacitated versions of these problems, building on previous
  work by \citet{frank-miklos}. We then build on this algorithm in two
  directions. First we develop a faster augmenting path subroutine for
  $k$-fold matroid union that, when appended to an approximation
  version of the push-relabel algorithm, gives a faster exact
  algorithm for some parameters of $k$. In particular we obtain a
  subquadratic-query running time in the uncapacitated setting for the
  three basic problems listed above.  We also obtain faster
  approximation algorithms for these problems with real-valued
  capacities by reducing to small integral capacities via randomized
  rounding.  To this end, we develop a new randomized rounding
  technique for base covering problems in matroids that may also be of
  independent interest.
\end{abstract}


\section{Introduction}

Matroids are a fundamental object of study in combinatorial
optimization. Three closely related and important problems involving
matroids are maximizing the size of the union of $k$ independent sets
(also known as \emph{matroid union} where we are taking $k$ copies of
the same matroid), computing $k$ disjoint bases in bases (a.k.a.\
\emph{matroid base packing}), and covering the groundset by $k$ bases
(a.k.a.\ \emph{matroid base covering}). These problems generalize
naturally to integral and real-valued capacities as we explain later.
This work develops faster exact and/or approximation problems for
these and some other closely related problems.  The main thrust of our
improvements comes from developing a faster and unifying
\emph{push-relabel} algorithm for the integer-capacitated versions of
these problems.

The push-relabel framework is most commonly associated with maximum
flow. The push-relabel framework for flow, introduced by
\cite{Goldberg1985} and then improved and generalized by
\cite{goldberg-tarjan-push-relabel} (cf.\ \cite{goldberg-87}), is an
elegant framework that has lead to new algorithms and perspectives for
max-flow (e.g.,
\cite{ahuja-orlin-89,CheriyanMaheshwari1989,Tuncel1994}) and related
problems such as minimum cost flow \cite{GoldbergTarjan1990},
parametric flow \cite{ggt-89} and directed minimum cuts
\cite{hao-orlin,hrg-00} (to name a few). Empirical work has also shown
that push-relabel algorithms can work well in practice (e.g.,
\cite{cg-97,cgm-98}).  Loosely speaking, the push-relabel algorithm
for flow departs from previous algorithms by making local
improvements, ``pushing'' flow along one edge at a time, rather than
augmenting along paths. Rather than maintaining an $(s,t)$-flow
(conserving flow at non-terminals), the push-relabel algorithm
maintains a relaxation of flow called \emph{preflows} the allow for
positive surplus at non-terminal vertices. Vertex labels, assigning
integer levels to each vertex, are introduced to guide the push
operations and obtain polynomial running times.

The push-relabel concept has been extended (in a much more abstract
form) to more general combinatorial problems including submodular flow
\cite{FujishigeZhang1996}, submodular intersection
\cite{FujishigeZhang1992}, submodular minimization
\cite{FleischerIwata2000,FleischerIwata2003,IwataOrlin2009}, and other
related problems
\cite{IwataMurotaShigeno1997,IwataMcCormickShigeno2000}.  More
recently, \citet{frank-miklos} developed simpler strongly
polynomial-time push-relabel algorithms for abstract combinatorial
optimization problems ranging from matroid partition to submodular
flow. Their work builds on and helps unify some of these past
results. Our work directly builds on \cite{frank-miklos}. Focusing on
the class of matroid problems listed above, we contribute both
structural observations and algorithmic techniques to their framework
that accelerate exact algorithms and also extend the framework to
yield fast approximation algorithms.

We complement these techniques (that fall within the push-relabel
framework) in two different ways. The first is to design an augmenting
path subroutine specialized to integer-capacitated $k$-fold matroid
union, which can be used augment an approximate solution produced
(very quickly) by the push-relabel framework to an optimal one. We
point out that this augmenting-path subroutine requires additional
structure provided by the push-relabel framework to work. The
augmenting paths lead to faster running time for certain ranges of
parameters, including a subquadratic-query running time for the
uncapacitated setting.  The second extension is based on random
sampling, which is used to reduce approximation problems for
real-capacitated problems to approximation problems with (small)
integer capacities. Some of these randomized rounding techniques are
new while others are provided by previous work \cite{Karger1998}.  The
reduced setting with integer capacities is particularly well-suited
for the push-relabel algorithms that we develop and this leads to a
series of approximation algorithms with nearly linear oracle and
running time complexity for real-valued capacities.

\subsection{Outline of results}

\labelsection{matroid-results}
\labelsection{graph-results}
\label{results}
\labelsection{results}

We now outline the improved running times obtained in this work.  For
the sake of brevity we defer more detailed descriptions of each
problem to later in the paper when they are analyzed in full. For each
we state the new running time and only the most competitive and
directly comparable running times in the literature (that we are aware
of). We provide pointers to the full theorem statements for each
result. Additional background is described in
\cref{additional-background}.  The problems we discuss are all basic
and well-studied problems in matroid optimization and can be found in
\cite{schrijver-book}, which we refer to for additional background. We
pose these problems for general matroids in the oracle model and for
the graphic matroid (i.e., forests of an undirected graph).

We briefly mention some preliminaries needed to describe the results;
additional preliminaries including further notation and relevant
definitions are provided in \refsection{preliminaries}.  For matroid
problems, we let $n$ denote the number of elements, and $\therank$ the
rank of the matroid. We assume access to an \emph{independence oracle}
to which we can query if a set of elements $S$ is independent. We
adopt the standard and simplifying convention of counting, as part of
the independence query, the $\bigO{\sizeof{S}}$ work one typically
needs to assemble and transmit the set $S$ to the oracle. For graph
problems, we let $m$ denote the number of (distinct) edges and $n$ the
number of vertices in the graph. When the graph has integer edges
capacities, we let $U$ denote the total capacity.

The results come in one of three flavors: (1) exact algorithms for
integral capacities that produce integral solutions, (2) integral
approximation algorithms for integral capacities that produce integral
solutions, and (3) approximate decision algorithms for real-valued
capacities. Type (1) usually comes with two (incomparable) running
times --- one given directly by the push-relabel algorithm, and one
combining push-relabel with augmenting paths.  Algorithms of the
second type are typically truncated versions of the corresponding
push-relabel algorithms of type (1).  Algorithms for the third type
are typically obtained by reducing to problems of the second type via
random sampling. All the algorithms for integer capacities --- of type
(1) and type (2) --- are deterministic and construct integral and
mutually certifying primal and dual solutions. The algorithms for
real-valued capacities (type (3)) are randomized Monte Carlo
algorithms that succeed with high probability. They construct an
approximate and integral dual solution but not a primal solution.

We proceed to present the results. We will present the results for
matroids in the independence oracle model first and then a parallel
set of results for the graphic matroid.

In the uncapacitated setting, the \emph{$k$-fold matroid union}
problem asks for $k$ bases $\Bs$ maximizing the size of their
union. This naturally extends to integer capacities where we now count
each element in $\Bs$ with multiplicity up to its capacity. We first
obtain a running time of $\MatroidUnionTime$ independence queries,
where $\opt$ denotes the optimum value, via the push-relabel algorithm
(\reftheorem{matroid-union}). We also obtain a running time of
\begin{math}
  \MatroidUnionWithAugmentingPathsTime
\end{math}
independence queries, where
$\MatroidUnionWithAugmentingPathsNotation$, by combining the
push-relabel algorithm with augmenting paths
(\cref{matroid-union-with-augmenting-paths}). Note that in the
uncapacitated setting, $\opt \leq n$ and $n' \leq n \log{r}$, so the
second running time is at most $\apxO{n^{3/2}}$ queries for any choice
of parameters $k$ and $\therank$.  These running times are to be
compared to $\CunninghamMatroidUnionTime$ for the uncapacitated
setting \cite{Cunningham1986}, where $Q$ denotes the time for a query
to an independence oracle. One can also reduce uncapacitated $k$-fold
matroid union to unweighted matroid intersection (with $nk$ elements
and rank $\opt$); this yields a randomized running time of
$\apxO{n k \sqrt{\opt}}$ queries \cite{Blikstad2021}.

We also consider approximation algorithms for capacitated $k$-fold
matroid union. Let $\eps \in (0,1)$ be a given parameter; the goal is
to achieve an objective value of at least $\epsless \opt$. For integer
capacities we obtain a running time of $\ApxMatroidUnionTime$
independence queries (\cref{apx-matroid-union}). The algorithm
produces both integral primal and dual solution which mutually certify
that they are $\epspm$-approximately optimal. As per comparable
running times, while \cite{Cunningham1986} does not consider
approximations explicitly, \cite{Cunningham1986} implicitly contains
an $\epsless$-approximation algorithm for the uncapacitated setting
with time $\bigO{knQ + (\opt n Q + k n) / \eps}$. One can also reduce
uncapacitated $k$-fold matroid union to approximate matroid
intersection, yielding a running time of
$\apxO{n k \sqrt{\opt} / \eps}$ queries \cite{Blikstad2021}.  The
$k$-fold matroid union problem generalizes to real-valued capacities
where one instead optimizes over fractional combinations of $k$
bases. For this setting we develop a randomized algorithm with
randomized running time bounded by $\RandomizedApxMatroidUnionTime$
independence queries. We are not aware of comparable algorithms in the
literature.

The next problem we consider is base packing, which we describe for
integer capacities. The goal is to compute $k$ bases so that no
element appears in more bases than its capacity.  It is easy to see
that (exactly solving) the problem reduces to $k$-fold matroid union,
hence we obtain the same running times as listed above. In addition to
the \cite{Cunningham1986} result mentioned above, one can also compare
to a Las Vegas randomized algorithm for the uncapacitated setting that
runs in randomized $\KargerBasePackingTime$ independence queries with
high probability \cite{Karger1998}. \cite{Karger1998} also gives a
deterministic algorithm running in
$\KargerDeterministicBasePackingTime$.

We also consider approximate base packing. The goal is to either pack
$\epsless k$ bases or certify that there is no packing of $\epsmore k$
bases, for a parameter $\eps \in (0,1)$. Here approximate $k$-base
packing does not reduce directly to approximate $k$-fold matroid
union, although we use similar techniques. We obtain a running time of
$\ApxBasePackingTime$ (\cref{apx-base-packing}). This is to be
compared to a randomized Monte Carlo algorithm for the uncapacitated
setting running in $\apxO{n + \therank^3 k / \eps^3}$ independence
queries \cite{Karger1998}. For real-valued capacities we obtain a
randomized Monte Carlo algorithm running in
$\RandomizedApxBasePackingTime$ independence queries that succeeds
with high probability (\cref{randomized-matroid-strength},
\cpageref{randomized-matroid-strength}). This can be compared to a
deterministic $\apxO{n k / \eps^2}$-query time algorithm from
\cite{cq-17-soda} or a randomized Monte Carlo algorithm running in
$\bigO{n + \therank^3 \log{kr}^{\bigO{1}} / \eps^5}$ independence
queries \cite{Karger1998}. We believe that ideas from
\cite{Karger1998,cq-17-soda} can be combined to also obtain a
randomized Monte Carlo algorithm running in
$\apxO{n + \therank / \eps^4}$ independence queries. Our randomized
algorithm for fixed $k$ can be converted to a randomized algorithm to
compute the maximum value $k$ (which we call the \emph{matroid
  strength}) via a modified binary search. The running time we obtain
is slightly better than one gets from a straightforward application of
binary search (\cref{search-matroid-strength},
\cpageref{search-matroid-strength}).

Next we discuss base covering. Here the goal is to compute $k$ bases
so that for each element, the number of bases containing each element
is at least the capacity of that element. If that is not possible then
one expects a dual certificate of feasibility. As with packing, exact
base covering can be solved by $k$-fold matroid union so we inherit
those running times. To the best of our knowledge, the best comparable
running time is that of the $k$-fold matroid union algorithm of
\cite{Cunningham1986} for the uncapacitated setting, mentioned
above. (In particular we are not aware of developments for base
covering analogous to \cite{Karger1998} or \cite{cq-17-soda} which
focus on base packing.)  For $\epsless$-approximate base covering with
integer capacities, which we note does not reduce to approximate
$k$-fold matroid union, we obtain a $\ApxBaseCoveringTime$ query
running time (\cref{apx-base-covering}). (Note that $n \leq k r$
without loss of generality for base covering.) For real-valued
capacities we may assume without loss of generality $k = 1$. In this
case we have the \emph{matroid membership} problem where the goal is
to decide if a real-valued value is in the independent set polytope of
a matroid. \cite{Cunningham1984} gave the first strongly-polynomial
time algorithm and there is an exact algorithm running in
$\bigO{n^3 \therank^2}$ independence queries \cite{Narayanan1995}.  We
obtain a randomized Monte Carlo algorithm running in
$\RandomizedApxCoveringTime$ queries
(\cref{apx-matroid-membership}). Again there is not much literature
explicitly an approximate base covering or matroid membership with
real capacities. However we believe that the techniques in
\cite{cq-17-soda} could have also obtained a deterministic
$\apxO{n k / \eps^2}$-query time algorithm; we note that such an
algorithm does not produce an integral packing when the capacities are
integral.

The final problem we discuss for matroids is optimal reinforcement. In
the integral version, given a matroid with integer element capacities,
real-valued element costs, and an integer parameter $k$, the goal is
to compute a minimum cost extension of the capacities so that the
resulting capacitated matroid has matroid strength
$k$. \cite{Cunningham1985c} proposed and analyzed this problem in
graphic matroids. Generalized to matroids, his algorithm is a
reduction to $\apxO{n}$ calls to $k$-fold matroid union. Multiplied
against the running time for \cite{Cunningham1986} mentioned above
(here $\opt = k \therank$) gives a running time of
$\CunninghamOptimalReinforcementTime$.  (\cite{Cunningham1985c} also
considers real capacities which we do not address.) We show how to
compute the optimum reinforcement in time equal to 1 call to our first
push-relabel algorithm for $k$-fold matroid union, running in
$\MatroidUnionTime$ independence queries. More precisely, we show that
the push-relabel algorithm directly solves the optimal reinforcement
with only a minor modification to how we initialize the algorithm.

We now outline the corresponding results for the graphic
matroid. Generally speaking, we take the algorithms above for general
matroids and replace the independence oracle with appropriate data
structures that can answer queries directly. This generally replaces
the query in the running times with a polylogarithmic overhead (or
better, such as the inverse Ackermann $\alpha(n)$.) For $k$-fold
matroid union --- that is, maximizing the total size of a packing of
$k$ forests --- we obtain a running time of $\ForestUnionTime$ by
push-relabel alone (\cref{max-forests}) and $\ForestUnionTimeB$ (in
connected graphs) by combining push-relabel with augmenting paths
(\cref{theorem:forest-union-b}). (Note that $\opt \leq nk$.)
Comparable running times are (a) $\bigO{n^2 k \log k}$ and (b)
$\bigO{\min{U, nk} \sqrt{k (U + n \log n)}}$.
\cite{GabowWestermann1992}. In particular our first running time
improves (a) and our second running time improves (b) because
$\opt \leq \min{U,nk}$.  For $\epsless$-approximations, we obtain an
$\ApxForestUnionTime$ for integer capacities and randomized
$\RandomizedApxForestUnionTime$ time for real capacities.

Packing bases translates to packing spanning trees. This can be solved
by the algorithms for $k$-fold union and we obtain the same running
times. In addition to the running times listed above for $k$-fold
matroid union, one can also compare to running times of
$\bigO{U n \log{m/n}}$ \cite{GabowWestermann1992} and
$\bigO{k n \sqrt{U + n \log n}}$ \cite{Gabow1991}. For
$\epsless$-approximations, we obtain $\ApxTreePackingTime$ for
approximate integer tree packings. This can be compared to a Las Vegas
randomized algorithm that runs in $\apxO{k n^{3/2} / \eps^2}$ time
with high probability.  The maximum number of spanning trees that can
be packed into a graph is called the network strength. For real-valued
capacities we obtain a randomized $\RandomizedApxTreePackingTime$ time
for $\epspm$-approximately testing the network strength and a slightly
greater running time for approximating the network strength up to a
$\epspm$-factor.  This result result can be compared with a
deterministic $\apxO{m / \eps^2}$ time algorithm or a randomized
$\apxO{m + n / \eps^4}$ time algorithm in \cite{cq-17-soda}.

Covering by bases corresponds to covering by spanning trees. Again new
exact running time for integer capacities are obtained via the
$k$-fold union algorithms listed above. Additional comparable running
times besides $k$-fold union are $\apxO{U^{5/3}}$ and
$\bigO{U n \log{n}}$ \cite{GabowWestermann1992}.  For
$\epsless$-approximations with integer capacities, we obtain a running
time of $\ApxTreeCoveringTime$. The minimum number of spanning trees
required to cover a graph is called the \emph{arboricity}. The
arboricity can be computed exactly in $\bigO{m n \log{n^2 / m}}$ time
\cite{Gabow1998}. We obtain a randomized
$\RandomizedApxTreeCoveringTime$ time Monte Carlo algorithm for
$\epspm$-approximately testing the arboricity and slightly greater
running time for approximating the arboricity up to an $\epspm$-factor
(\cref{apx-test-arboricity,apx-search-arboricity}).  As was the case
for matroids, there are not as many developments for approximate
covering by spanning trees as for packing spanning trees. However we
believe the techniques in \cite{cq-17-soda} lead to
$\apxO{m / \eps^2}$ and randomized $\apxO{m + n / \eps^4}$-time
algorithms to $\epspm$-approximate the arboricity (value).

Lastly, by the same reduction as for matroids, the optimal
reinforcement problem in graphs is solved by a single call to the
push-relabel algorithm for $k$-fold union, running in
$\ForestUnionTime$ time where $\opt$ refers to the optimum for
$k$-fold union, and is at most $nk$. This improves
\cite{Cunningham1985c}'s reduction to $n$ calls to $k$-fold union, as
well as the running time of $\bigO{n^2 m \log{n^2 / m}}$ by
\cite{Gabow1998} for small $k$. (\cite{Gabow1998}'s algorithm and
\cite{Cunningham1985c}'s reduction also extend to real-valued
capacities.)



\subsection{Discussion and overview of technical ideas}
\labelsection{techniques}
\label{techniques}

(The reader may want to skip this section on first pass and return to
it after obtaining some technical familiarity with the algorithms
presented in the body of this article.)

As mentioned above, our improved running times start with a common
``matroid push-relabel'' algorithm building on ideas from \FM. Before
describing our enhancements, it may be helpful to first describe their
framework at a high-level, especially given the unusual perspective
for those coming for flow.

Let us informally describe their algorithm for $k$-fold matroid union
for illustrative purposes. We assume the uncapacitated setting for
simplicity; the goal is to compute a set of $k$ bases $\Bs$ maximizing
the size of their union. Call an element $e$ \emph{uncovered} if it is
not in any $\basei$, \emph{covered} if it is in some $\basei$, and
\emph{overpacked} if it is in more than one $\basei$. We want to cover
as many elements as possible, and generally speaking, overpacked
elements represent wasted slots among the bases. The push-relabel
algorithm of \FM manipulates $\Bs$ by repeatedly selecting an
uncovered element $e$ and trying to exchange it into some $\basei$ so
that it is uncovered. To make direct progress one would have to
exchange it out for an overpacked element $d$; otherwise the size of
the union stays the same. However, such a profitable exchange may not
be available, even if $\Bs$ is not yet an optimal solution. While one
can exchange $e$ for other covered (but not overpacked) elements, it
is not clear how this helps. This is analogous to pushing flow from
one non-terminal vertex with surplus to another non-terminal vertex;
it is not clear that we are making progress towards a sink. This is
where the \emph{relabel} aspect of the push-relabel framework comes
in. Each element is labeled by an integer \emph{level}, and overpacked
elements are kept at level $0$. Exchanges are restricted so that an
element $e$ is exchanged into a $\basei$ for an element $d$ that is
one level below $\basei$. In some sense, the ``excess'' represented by
$e$ being uncovered shifts down one level to $d$, and thus closer to
the overpacked elements at level $0$. These ideas eventually lead to a
more elaborate argument in \FM about why the algorithm terminates in
polynomial time.

Above is a sketch omitting details, proofs of correctness, and even a
complete description of the algorithm (deferring a more technical
treatment to later).  However it starts to form an analogy between the
familiar push-relabel framework for max flow, and the abstract version
presented by \cite{frank-miklos}. Instead of labeling vertices of a
graph, we label elements of a matroid. In flow, we push flow along
edges from one vertex to another; with matroids, we ``push'' exchanges
of one element for another that maintain feasibility. Similar to flow,
pushes are restricted to go ``down'' a level, and when no pushes are
available, there is a \emph{relabel} operation where some elements
have their level increased. Doing so may reveal a violating constraint
induced by the level sets of vertices, similar to how minimum cuts
emerge from the labels in flow.

In analyzing their push-relabel framework, for problems ranging from
matroid partition to submodular flows, \cite{frank-miklos} focuses on
demonstrating that the algorithms are strongly polynomial while
keeping the algorithms and analysis as simple as
possible. (Historically, obtaining strongly polynomial running times
was highly non-trivial for these abstract problems.)  To this end,
rather than directly bound the running time, \cite{frank-miklos} gave
worst-case polynomial bounds on the number of ``basic operations'' ---
the number of push and relabel operations --- made by the
framework. (E.g., $\bigO{n^5}$ basic operations for the $k$-fold
matroid union problem above.) One can show that it takes polynomial
time (and queries to an independence oracle) to identify and execute a
basic operation, but we caution that this is far less straightforward
to do this than for flow, due to the abstract nature of matroids and
the oracle model. For example, in flow, the edges explicitly specify
where we can ``push'', while with matroids, finding an exchangeable
pair of elements $e_1$ and $e_2$ (as above) may require nested loops
over the ground set of elements and an independence query for each
inner iteration. There are additional technical issues that
\cite{frank-miklos} addresses which have no obvious analogy for flow.

We have taken to calling this framework \emph{matroid push-relabel},
to distinguish from push-relabel for flow.  Initially we were drawn by
the conceptual appeal of \cite{frank-miklos}, and started developing
algorithms for some more specific problems hoping at best for some
simpler or more practical alternatives to existing algorithms (similar
to the role now assumed by push-relabel for flow). Given the large
bounds and high level of abstraction in \cite{frank-miklos}, it was
not at all clear that competitive bounds could be obtained from
matroid push-relabel for basic, long-studied problems where there are
alternative approaches that seem more direct.  We were surprised to
discover that, upon developing several more ideas within the matroid
push-relabel framework, one can actually improve the best known bounds
for several of these problems.

As mentioned above, the push-relabel algorithm manipulates a
collection of bases and assigns levels to each element.  The algorithm
modifies the bases and levels while obeying a set of ``push-relabel
invariants'', proposed by \cite{frank-miklos}, which impose a
discipline on the exchanges made to the bases. %
We extend these invariants slightly by introducing an integer-valued
parameter called the \emph{height}. The height is the minimum level of
any uncovered element. One motivation for the height is to facilitate
the analysis of fast approximation algorithms, as discussed below. The
running time of our core matroid push-relabel algorithm is expressed
as a function of height. The height parameter is chosen based on the
problem and whether we seek exact or approximate solutions.  %

As mentioned above, the algorithms in \cite{frank-miklos} are not very
concrete (let alone efficient), and the algorithms were analyzed to
the point of bounding the number of push/relabel operations, as
opposed to bounding the running time. We fill in the running time
analysis and introduces several more ideas to improve the running
time. Some of these ideas are somewhat subtle, detailed and local;
such as a refined bound on the height for $k$-fold matroid union, or
more careful application of data structures for spanning tree
problems. There are also some broader ideas that are easier to isolate
and which we now highlight below.

\paragraph{Level-wise decreasing order of bases.} Recall that the
push-relabel algorithms maintains a collection of $k$ bases for an
input parameter $k$. Every time we want to exchange an element $e$
into the solution (so to speak), the multitude of bases raises an
algorithmic issue quickly identifying a suitable base in which to
exchange $e$. A naive approach loops through all the bases which is
very slow, especially with large convex combinations.  An important,
new idea introduced in this work, that seems very specific both to
matroids and the push-relabel framework, is to maintain the bases in
``level-wise decreasing order''. The definition of this ordering is
based on the upper level sets of the bases induced by the labels. We
require and maintain the bases in such an order that for every level,
the level set of one base spans the level set of the next
one. Speaking intuitively and abstractly, it turns out that the
monotonic nature of this order cooperates nicely with the way that
elements are relabeled and exchanged in the push-relabel framework. In
particular, some simple and necessary greedy rules for selecting
between different choices of push and relabel operations are shown to
be sufficient for maintaining the descending order. Moreover, given
the bases in level-wise descending order, we no longer have to loop
over the bases as (loosely) described above. Instead, we can apply a
binary search to identify the right base in logarithmic time and
queries.  The descending order also allows for a binary search along
the levels when relabeling an element, which also improves the running
time.

\paragraph{Approximation via truncation.} An important technique for
obtaining fast approximations in the push-relabel framework is the
idea that a truncated height, depending primarily on the desired
accuracy, suffices to obtain an approximate solution. This is
analogous to the well-known connection; in problems such as
edge/vertex disjoint paths, bipartite matching, and matroid
intersection; between the length of the shortest augmenting path and
the quality of the current solution (e.g.,
\cite{hopcroft-karp,Cunningham1986,even-tarjan}).  In this work, the
matroid push-relabel algorithms frame matroid optimization problems in
a perspective more amenable to these types of arguments. For many
problems, we show that either $\bigO{1/\eps}$ or
$\bigO{\log{n} / \eps}$ levels suffice to obtain a $\epspm$-factor
approximation. We note that there are different arguments that lead to
either the $\bigO{1/\eps}$ or $\bigO{\log{n} / \eps}$ bounds, and some
additional analysis was required to identify which was appropriate for
each problem.

As mentioned earlier, besides the enhancements to the matroid
push-relabel algorithm, we develop two more techniques that enhance
the applicability of the push-relabel algorithm.

\paragraph{Randomized rounding}
To extend the push-relabel framework for integer capacities to
real-valued capacities, for the sake of fast approximation algorithms,
we employ randomized rounding. More specifically, we use random
sampling to discretize the capacities and effectively reduce the
capacities to small integer values (at most $\bigO{\ln{n} / \eps^2}$).
For packing problems, the randomized techniques we need are already
provided by \cite{Karger1998}. For $k$-fold matroid union and covering
problems, such techniques were not known, and the arguments from
\cite{Karger1998} did not seem to extend. We develop a new analysis
for these remaining problems. Interestingly, this analysis also
recovers the results of \cite{Karger1998}, via an arguably simpler
proof. (In particular, the new analysis does not depend on the random
contraction algorithm.)  See \cref{general-sparsification},
\refsection{sparsification}. The randomized techniques for matroid
base covering and $k$-fold matroid union may be of independent
interest.

\paragraph{Augmenting paths.} Historically the most common approach to
the problems considered here is via augmenting paths.  Augmenting
paths can also be used to extend an approximate (integral) solution to
an optimal one (which is our application).  With augmenting paths for
(say) $k$-fold matroid union, one maintains a packing of $k$
independent sets $\Is$ and tries to extend it one element at a
time. Extending such a packing is non-trivial, and may require a
complicated sequence of exchanges to open up a slot for a new element,
so to speak. A shortest path between a designated source and sink in
this graph can be shown to give an augmenting path.  Now, the most
straightforward approach is to build out the entire auxiliary graph
explicitly by testing for the presence of each possible arc, and then
running BFS in the resulting graph. However the graph is potentially
dense and building out the graph becomes the bottleneck.

We instead explore techniques that look for the desired path
implicitly. At a high-level, one recognizes that identifying all the
vertices reachable from the designated source does not require
exploring all the arcs. Indeed, we need not test the existence of an
arc where the head is already ``marked'' as explored. That said, we do
need to be able to identify arcs to ``unmarked'' auxiliary
vertices. To help us search for these useful arcs we impose additional
invariants. One invariant is to keep $\Is$ in ``decreasing order'', in
a way similar to the level-wise order of decreasing bases in the
matroid push-relabel algorithm. This
restricts the family of augmentations we allow in each iteration. The
second invariant comes within a single search for an augmenting
path. As we mark auxiliary vertices as they are explored, we also
require the subset of elements in $\Is$ corresponding to marked
auxiliary vertices to also be in ``decreasing order''. This is
addressed by introducing a ``pre-search'' subroutine that is called on
an auxiliary vertex before marking it. With this additional structure
--- keeping the $\Is$ and the marked elements of the $\Is$ (so to
speak) in descending order --- it becomes much easier to navigate the
auxiliary graph implicitly. The overall running time for one search
comes out to roughly a logarithmic number of independence queries per
vertex in the auxiliary graph. (\Cref{search-augmenting-path},
\refsection{augmenting-paths}.)

We use the augmenting path subroutine to extend an
$\epsless$-approximate solution produced by the push-relabel algorithm
to an optimum solution, for an appropriate choice of $\eps$. However,
to apply our augmenting path algorithm we require the output of the
push-relabel algorithm to satisfy certain invariants: namely, the
bases $\Bs$ should contain a packing $\Is$ of the same total capacity
and in descending order. Fortunately this is the case, and the proof
critically depends on the fact that the $\Bs$ were in level-wise
decreasing order already.

We note that some similar ideas pertaining to implicitly navigating
the auxiliary graph have been applied recently to accelerating
augmenting path algorithms for matroid intersection
\cite{Nguyen2019,CLSSW2019,BBMN2021,Blikstad2021}, which is a closely
related problem. Still, additional ideas specific to $k$-fold matroid
union as well as the added structure from the matroid push-relabel
algorithm are required to obtain our running time.

\paragraph{Future work.}
We recognize that some of the ideas here can be useful for exact
push-relabel algorithms for real capacities and defer this to future
work.  We have also continued to develop algorithms from the
push-relabel perspective for more abstract problems such as
polymatroid intersections and submodular flow. These topics require
much more abstract machinery and the conceptual focus is different
from the presentation here which is more specialized to matroids. We
defer these developments to future work.


\subsection{Organization.}
The rest of this work is organized as follows.
\begin{itemize}
\item In \cref{preliminaries} we present preliminary definitions and
  notation.
\item In \refsection{impr}, we present and analyze the matroid
  push-relabel algorithm for $k$-fold matroid union with integer
  capacities.
\item In \refsection{base-packing-covering}, we analyze the base
  packing and covering problems for integer capacities.
\item In \refsection{spanning-trees}, we implement the matroid
  push-relabel algorithms for the graphic matroid.
\item In \refsection{reinforcement} we analyze the minimum cost
  reinforcement problem.
\item In \refsection{apx-capacities} we develop the randomized
  algorithms for real-valued capacities.
\item In \refsection{augmenting-paths} we develop the augmenting path
  algorithm.
\item Additional background is given in \ref{additional-background}.
\end{itemize}


\paragraph{Acknowledgements.} We thank Chandra Chekuri for helpful
feedback. We thank the reviewers for helpful feedback and additional
pointers to references.

\section{Preliminaries}

\label{preliminaries}
\label{prelim}

\labelsection{preliminaries}

We briefly introduce matroids and refer to \cite{schrijver} for
additional background. A matroid
$\matroid =(\groundset, \independents)$ consists of a finite ground
set $\groundset$ and a collection of \emph{independent sets}
$\independents \subseteq 2^\groundset$ that satisfy three properties:
(i) $\emptyset \in \independents$ (ii) $A \in \independents$ and
$B \subset A$ implies $B \in \independents$ and (iii)
$A, B \in \independents$ and $|B| > |A|$ implies that there is
$e \in B \setminus A$ such that $A \cup \{e\} \in \independents$. The
$\rank$ function of a matroid is an integer valued function over the
subsets of $\groundset$ where $\rank{S}$ is the cardinality of the
largest independent set contained in $S$. We let
$\spn{S} = \setof{e \where \rank{S + e} = \rank{S}}$ denote the set of
elements spanned by $S$.  For an independent set
$I \in \independents$, and an element $e \in \spn{I} \setminus I$,
there is a unique minimal set in $I+ e$, called the \defterm{circuit}
of $e$ in $I$ and denoted $\circuit{I + e}$.

The \emph{independence polytope} is the set of vectors
$x \in \nnreals^{\groundset}$ that can be expressed as a convex
combination of indicator vectors of independent sets. We let $\ip$
denote the independence polytope. For $k > 0$, we let $k \ip$ scale up
$\ip$ by a factor of $k$; equivalently, $k \ip$ denotes the the set of
vectors $x$ such that $x / k \in \ip$.

In all our problems the elements are equipped with \emph{capacities}
$\capacity \in \preals^{\groundset}$. The capacities are usually
integral except in \refsection{apx-capacities} where we consider
real-valued capacities. Abusing notation, for a set
$S \subseteq \groundset$, we denote the sum of capacities over $S$ by
\begin{align*}
  \capacity{S} \defeq \sum_{e \in S} \capacity{e}.
\end{align*}

\label{def-num-covered}

Many of our problems manipulate a set of $k$ bases $\Bs$ which may
have overlapping elements. In such a context, for an element
$e \in \groundset$, let
\begin{align*}
  \numCovered{e} \defeq \sizeof{\setof{i \in [\numB] \where e \in B_i}}
\end{align*}
We say that an element is \emph{uncovered} if
$\numCovered{e} < \capacity{e}$, \emph{covered} if
$\numCovered{e} \geq \capacity{e}$, \emph{feasibly packed} if
$\numCovered{e} \leq \capacity{e}$, and \emph{overpacked} if
$\numCovered{e} > \capacity{e}$.


\section{Matroid push-relabel with integer capacities}

\labelsection{impr}%

\newcommand{\aux}[2]{#1^{\smash{(#2)}}}%

This section presents the matroid push-relabel algorithm. Our discussion
centers on the $k$-fold matroid union problem, presenting an algorithm
that accelerates an algorithm for $k$-fold matroid union presented in
\cite{frank-miklos}. The $k$-fold matroid union problem was briefly
introduced in \cref{results} and we reintroduce the problem here.

The input consists of a matroid
$\matroid = (\groundset, \independents)$, integer capacities
$\capacity: \groundset \to \naturalnumbers$, and an integer $k$.  This
input defines the following dual min-max problems shown by
\cite{NashWilliams1967} to have equal objective values:
\begin{gather}
  \text{maximize } \sum_{e \in \groundset} \min{\capacity{e},
    \numCovered{e}} \text{ over $k$
    bases }
  B_1,\dots,B_k \in \independents; \labelthisequation{max-imp} \\
  \text{minimize } k \rank{S} + \sum_{e \in \bar{S}} \capacity{e} \text{ over all sets
  } S \subseteq \groundset. \labelthisequation{min-imp}
\end{gather}
($\numCovered{e}$ is defined in \cref{preliminaries}.)
\refequation{max-imp} is called the \emph{$k$-fold matroid union}
problem; we refer to \refequation{min-imp} simply as its dual problem.
As mentioned earlier, \cite{Cunningham1986} gave a
$\bigO{(\opt^{3/2} + k) n Q + \opt^{1/2} k n}$-time algorithm for the
uncapacitated version (i.e., $\capacity{e} = 1$ for all $e$) of the
problem. Additionally, \FM gave a bound of $\bigO{n^5}$ ``basic
operations'' (which are defined below) in the uncapacitated
setting. We note that both \cite{Cunningham1986,frank-miklos} consider
the more general matroid union setting, where each $B_i$ is a base in
a different matroid (over the same ground set).

We first introduce the high-level components of the matroid
push-relabel framework, based on \FM, in \cref{imp-framework}. We then
analyze how optimality is obtained with the framework, improving
bounds in \FM as well as extending the analysis to approximations, in
\cref{imp-analysis}. Finally we present and analyze the faster
algorithm in \cref{imp-algorithm}.

\subsection{Components of the matroid push-relabel framework}

\label{imp-framework}

The matroid push-relabel maintains a map $\deflevels$ assigning
\emph{levels} to each element, initially set uniformly to $0$. The
algorithm also maintains $k$ bases $B_1,\dots,B_k \in \bases$, all
initialized to be the same base chosen arbitrarily.

Our discussion frequently groups elements by their levels and to this
end it is convenient to introduce the following notation.  For a fixed
level $j$, we let $\level_j$ denote the set of elements at level
$j$. We also let $\level_{\leq j}$ denote the set of elements at level
less than or equal to $j$; similarly we have $\level_{< j}$,
$\level_{> j}$, and $\level_{\geq j}$.

As the algorithm updates the bases $\Bs$ by inserting elements, for
each base $\basei$ and each element $e \in \basei$, the algorithm
tracks the level of $e$ when it was inserted into $\basei$. For
$i \in [k]$ and $j \in \nnintegers$, we let
$\basei_j \subseteq \basei$ be the subset of elements $e$ in $\basei$
that were inserted into $\basei$ when $e$ was at level $j$. (Note that
$\basei_j$ does \emph{not} equal $\basei \cap \level_j$.) We let
$\basei_{\geq j}$ denote the set of elements $e$ inserted into
$\basei$ when $\level{e}$ was at least $j$. Similarly we have
$\basei_{\leq j}$, $\basei_{> j}$, and $\basei_{< j}$.

In the matroid push-relabel algorithm, the $k$ bases $B_1,\dots,B_k$
form a candidate solution for the maximization problem
\refequation{max-imp}. The sub-level sets $\level_{\leq j}$ (where
$j \in \naturalnumbers$) represent candidate solutions for the
minimization problem \refequation{min-imp}.    The framework is designed
to ensure that
\begin{align*}
  \sum_{e \in \groundset} \min{\capacity{e}, \numCovered{e}} %
  \leq                                                       %
  k \rank{\level_{\leq j}} + \capacity{\level_{> j}}
\end{align*}
for all $j$. As shown in \cref{imp-optimality} below, at termination,
the algorithm identifies a level $j$ for which the inequality above is
(exactly or approximately) tight.  Thereby $B_1,\dots,B_k$ and
$\level_{\leq j}$ certify one another to be (exactly or approximately)
optimal for their respective problems.

The matroid push-relabel algorithm obeys the following invariants
proposed by \cite{frank-miklos}.  Here we say that an element $e$ is
\defterm{covered} if it is contained in at least $\capacity{e}$ bases,
and otherwise \defterm{uncovered}.
\begin{invariants}
\item \label{imp-disjoint}\label{imp-level-0} $\level{e} = 0$ for all
  elements $e$ that are in (strictly) more than $\capacity{e}$ bases.
\item \label{imp-span} For $i = 1,\dots,k$, and all levels $j$,
  $\basei_{\geq j}$ spans $\level_{> j}$.
\item \label{imp-covered} \label{imp-uncovered} All uncovered elements
  $e$ have $\level{e} \geq \Height$ for a parameter
  $\Height \in \nnintegers$.
\end{invariants}
Given a configuration of bases and level assignments, the
\defterm{height} is defined as the minimum level of any uncovered
element. That is, the height is the minimum value of $h$ satisfying
\cref{imp-covered}.

\newcommand{\IMPInvariants}{\cref{imp-disjoint,imp-span,imp-covered}\xspace} %

The matroid push-relabel algorithm is composed of essentially two
operations, which \FM calls \emph{basic operations}.
\begin{enumerate}
\item \emph{Push (exchange):} Given an uncovered element $e$, a base
  $B_i$, and an element $d \in B_{i,\level{e} - 1}$, such that
  $B_i - d + e \in \independents$, replace $B_i$ with
  $B_i - d + e$.\footnote{Note that $d$ may be the same as $e$,
    inserted earlier when $e$ was at a lower level.}
\item \emph{Relabel:} Given an uncovered element $e$, increase
  $\level{e}$ to $\level{e} + 1$.
\end{enumerate}

To preserve \cref{imp-span}, an uncovered element $e$ can only be
relabeled if there is no push operation available. As mentioned in
\refsection{techniques}, identifying a feasible push for an element
$e$ is a computational bottleneck, in contrast to
flow. \cite{frank-miklos} proved that the push and relabel operations
preserve \IMPInvariants which we assume as a fact.

\subsection{Optimality via the matroid push-relabel invariants}
\label{imp-analysis}\label{imp-optimality}

We will eventually develop an algorithm that tries to obtain, as
quickly as possible, a configuration of bases and levels satisfying
\cref{imp-disjoint,imp-span,imp-covered} above for a given height
parameter $\Height \in \naturalnumbers$.  First we show how particular
values of $\Height$ correlate with good solutions to the dual min-max
problems in \cref{equation:max-imp,equation:min-imp}. Here we have
analyses for both exact and approximate solutions.

\subsubsection{Exact solutions} We first consider exact solutions to
\cref{equation:max-imp,equation:min-imp}. Previously,
\cite{frank-miklos} showed that height $\bigOmega{n}$ suffices to
derive an exact solution in the more general setting of matroid union
(with different matroids). For the specific case of $k$-fold matroid
union we have the following stronger bound of $\therank + 2$.
\begin{lemma}
  \labellemma{imp-exact-analysis}\label{imp-exact-analysis} Suppose $\firstB,\dots,\lastB$ and
  $\level: \groundset \to \nnreals$ satisfy the invariants
  \IMPInvariants with height $\Height > \therank + 2$.  Then there
  exists a level $j$ such that
  \begin{align*}
    \sum_{e \in \groundset} \min{\capacity{e}, \numCovered{e}}
    \geq                        %
    \capacity{\level_{< j}} + k \rank{\level_{\geq j}}.
  \end{align*}
  This certifies that $\firstB \cup \cdots \cup \lastB$ is a maximum
  solution and that $\level_{\geq j}$ is a minimum dual solution.
\end{lemma}

\begin{proof}
  As a function of $j \in \nnintegers$, $\rank{\level_{\geq j}}$ is
  integral and nondecreasing from $0$ to $\therank$.  By the
  pigeonhole principle, there exists a level
  $j \in \setof{1,\dots,\Height}$ such that
  \begin{math}
    \rank{\level_{\geq j}} = \rank{\level_{> j}}.
  \end{math}
  We will prove the claim for this choice of $j$. By \cref{imp-span}),
  $\basei{\geq j}$ spans $\level_{> j}$, hence
  \begin{align*}
    \sum_{i=1}^{\numM} \sizeof{\basei{\geq j}}
    \geq k \rank{\level_{\geq j}} = k \rank{\level_{> j}}.
  \end{align*}
  By \cref{imp-disjoint}, no element in $\groundset_{\geq j}$ is
  overpacked, so
  \begin{align*}
    \sum_{i=1}^{\numM} \sizeof{\basei{\geq j}}
    \leq                        %
    \sum_{i=1}^{\numM} \sizeof{\basei \cap \groundset_{\geq j}}
    =
    \sum_{e \in \groundset_{\geq j}} \min{\numCovered{e}, \capacity{e}}.
  \end{align*}
  Lastly,
  by \cref{imp-covered}, all elements in $\groundset_j$ are covered, so
  \begin{align*}
    \capacity{\level_{< j}}
    =                           %
    \sum_{e \in \groundset_{< j}} \min{\capacity{e}, \x{e}}.
  \end{align*}
  Altogether we have
  \begin{align*}
    \capacity{\level_{< j}} + \numM \rank{\level_{\geq j}}
    \leq                        %
    \capacity{\level_{< j}} + \sum_{i=1}^{\numM} \sizeof{\basei{\geq
    j}}
    \leq                        %
    \sum_{e \in \groundset} \min{\capacity{e}, \x{e}},
  \end{align*}
  as desired.
\end{proof}

\subsubsection{Approximate solutions}

We now turn to approximations. Here we show that height $\bigO{\reps}$
suffices to obtain an $\epsmore$-approximations for both the primal
and dual problems.
\begin{lemma}\label{imp-apx-analysis}
  \labellemma{imp-apx-analysis} Let $\eps \in (0,1)$.  Suppose
  $\firstB,\dots,\lastB$ and $\level: \groundset \to \nnreals$ satisfy
  \IMPInvariants with height $\Height > \reps + 2$.  Then there exists
  a level $j$ such that
  \begin{align*}
    \capacity{\level_{< j}} + k \rank{\level_{\geq j}}
    \leq
    \epsmore
    \sum_{e \in \groundset} \min{\capacity{e}, \numCovered{e}}.
  \end{align*}
  This certifies that $\firstB,\dots,\lastB$ is a
  $1/\epsmore$-approximately maximum solution, and that
  $\level_{\geq j}$ is a $\epsmore$-approximately minimum dual
  solution.
\end{lemma}

\begin{proof}
  Since $\Height > \reps + 2$, there exists an index
  $j \in \setof{1,\dots,\Height-1}$ such that
  \begin{align*}
    \sizeof{\firstB_j}+ \cdots + \sizeof{\lastB_j} \leq %
    \eps \parof{\sizeof{\firstB{\geq 1}} + \cdots +
    \sizeof{\lastB{\geq 1}}}.
    \labelthisequation{apx-imp-small-level}
  \end{align*}
  We will prove the claim for this choice of $j$. Since each
  $\basei{\geq j}$ spans $\groundset_{> j}$ (per \ref{imp-span}),
  \begin{align*}
    k \rank{\level_{> j}}
    \leq                        %
    \sum_{i=1}^k \sizeof{\basei{\geq j}}.
  \end{align*}
  Additionally, by choice of $j$ per \cref{equation:apx-imp-small-level}, we have
  \begin{align*}
    \sum_{i=1}^k \sizeof{\basei{\geq j}}
    \leq                        %
    \sum_{i=1}^k \sizeof{\basei{> j}} + \eps \sum_{i=1}^k
    \sizeof{\basei{\geq 1}}
    \leq \x{\level_{> j}} + \eps \x{\level_{\geq 1}}.
  \end{align*}
  Now, since $\x{e} \leq \capacity{e}$ for all $e \in \level_{> 0}$
  and $\capacity{e} \leq \x{e}$ for all $e \in \level_{< \Height}$ by
  \cref{imp-covered,imp-disjoint}, we have
  \begin{align*}
    \capacity{\level_{\leq j}} + k \rank{\level_{> j}}
    \leq                        %
    \capacity{\level_{\leq j}} + \x{\level_{> j}} + \eps
    \x{\level_{\geq 1}}
    \leq                        %
    \epsmore \sum_{e \in \groundset} \min{\capacity{e}, \numCovered{e}},
  \end{align*}
  as desired.
\end{proof}

\subsection{A faster matroid push-relabel algorithm}

\label{impr-algorithm} \label{imp-algorithm}

Having now established how the push-relabel \IMPInvariants imply exact
or approximate, we turn to the algorithmic question of computing a
configuration that satisfies the invariants for a prescribed value of
height $\Height$.  The running times we obtain, as a function of
$\Height$, are described in the following theorem.  Below, we let
$\opt$ denote the common optimum value of
\cref{equation:max-imp,equation:min-imp}.
\begin{theorem}\label{impr}%
  \labeltheorem{impr} Given a matroid
  $\matroid = (\groundset, \independents)$ and parameters
  $k,\Height \in \naturalnumbers$, there is an algorithm that, in
  running time bounded by $\bigO{n + \Height \opt \log{k \therank}}$
  calls to an independence oracle, computes bases $\Bs$ and levels
  $\deflevels$ that satisfy \IMPInvariants with height $\Height$.
\end{theorem}

As a point of comparison, \FM showed that $\bigO{n \Height}$ basic
operations suffice to obtain height $\Height$ for the unweighted
setting (which does not account for other computational factors such
as identifying basic operations). The rest of this subsection is
devoted to proving \reftheorem{impr}.

As mentioned above, the algorithm initializes $\Bs$ to be any
arbitrary base. Here the greedy algorithm can compute a base in time
proportional to $\bigO{n}$ independence queries. Initially we set
$\level{e} = 0$ for all elements $e$.  In addition to \IMPInvariants,
we impose the following monotonicity condition on $\Bs$ and $\level$.

\begin{definition}
  \label{def:decreasing-bases}\label{decreasing-bases}
  Let $\Bs$ be a sequence of independent sets, and let $\deflevels$
  assign integer levels to each element. We say that $\Bs$ is
  \defterm{monotone decreasing} (or just \defterm{decreasing}) if
  $\basei{\geq j}$ spans $\basei[i+1]{\geq j}$ for all indices
  $i = 1,\dots, \numM - 1$ and all levels $j \in \nnintegers$.
\end{definition}

Given $\Height$, the goal of the algorithm is to reach a configuration
where $\level{e} \geq \Height$ for all uncovered elements $e$. To this
end, the push-relabel algorithm (described by \FM) repeatedly selects an uncovered
element $e$ with $\level{e} < h$, and in principle, wants to either
push $e$ into some $\basei$ to cover $e$, or relabel $e$ and bring
$\level{e}$ closer to $\Height$. For the faster algorithm, we describe
a new procedure, called \emph{greedy insertion}, that employs binary
search along both the levels and the bases, to more aggressively place
$e$ in the first available base (so to speak). Some justification is
required to argue that this process simulates a legal sequence of
push-relabel operations; we prove this after describing the procedure.
\begin{quote}
  \emph{Greedily inserting an element $e$:} Let $e \in \groundset$ be
  uncovered with $\level{e} < \Height$.
  \begin{enumerate}
  \item If $e$ is spanned by $\lastB{\Height - 1}$ then set
    $\level{e} = \Height$ and return.
  \item \label{imp-greedy-1} Otherwise identify the first level $j$
    such that $\lastB{> j}$ does not span $e$.
  \item \label{imp-greedy-2} Otherwise search for the first index $i$
    such that $B_{i,>j}$ does not span $e$. Set $\level{e} = j + 1$,
    and exchange $e$ into $B_{i}$ for an element $d \in B_{i,j}$ such
    that $B - d + e \in \independents$.
  \end{enumerate}
\end{quote}

All put together, the overall algorithm is as follows. We initially
set all bases $\Bs$ to an arbitrary base, and $\level{e} = 0$ for all
$e$. As long as there is an uncovered element $e$ with
$\level{e} < h$, we greedily insert $e$, which either places $e$ in a
base, or sets $\level{e} = h$.

It remains to analyze both the correctness and the running time of
this algorithm. We start with correctness, and in particular, we first
show that \algo{greedy-insertion} maintains \IMPInvariants. The proof
will require the fact that the bases were in decreasing order prior to
greedy insertion; hence we also prove that greedy insertion maintains
the decreasing order of bases.

\begin{lemma}
  \labellemma{impr-greedy-update} Suppose $\Bs$ and $\deflevels$ satisfy
  \IMPInvariants, and $\Bs$ are in descending order.  Then greedily
  inserting an element $e$ maintains \IMPInvariants and keeps the
  bases in decreasing order.
\end{lemma}
\begin{proof}
  We first show that greedy insertion maintains the descending order.
  We need only consider the case where we execute an exchange in step
  \ref{imp-greedy-2}, as otherwise there is no change to the bases.
  Thus, suppose we exchange an element $e$ into a base $B_i$ at level
  $j$, in exchange for an element $d$ at level $j-1$. We let $\Bs$
  denote the bases before the exchange. We let $B_i' = B_i - d + e$
  denote the updated base after the exchange; this is the only change
  to the sequence of bases.  It suffices to compare $B_{i}'$ to the
  bases $B_{i-1}$ and $B_{i+1}$ that precede and succeed $B_i$
  (assuming $B_i$ is not the first or last base, respectively). There
  is no need to verify levels $j'$ strictly larger than $j$ since
  these level sets do not change.

  Consider first $B_{i-1}$ (when $i > 1$). For any level $j' \leq j$, we have
  \begin{align*}
    B_{i, \geq j'}' \subseteq B_{i, \geq j'} \cup \setof{e}
    \tago{\subseteq} \spn{B_{i-1, \geq j'}},
  \end{align*}
  and taking the span of both sides gives the desired subset
  inequality.  Here \tagr is because
  $B_{i, \geq j'} \subseteq \spn{B_{i-1, \geq j'}}$ by the fact that
  $\Bs$ is in descending order, and also because
  $e \in \spn{B_{i-1,\geq j}}$ by choice of the index $i$.

  Consider now $B_{i+1}$ (when $i < \numM$).  At level $j$, we have
  \begin{align*}
    B_{i + 1, \geq j}       %
    \tago{\subseteq}
    \spn{B_{i, \geq j}}
    \tago{\subseteq}
    \spn{B_{i,\geq j}'},
  \end{align*}
  and taking the span of both sides gives the desired subset
  inequality. Here \tagr is by the existing descending order. \tagr is
  because $B_{i, \geq j}' = B_{i, \geq j} + e$.  For levels $j' < j$,
  we have
  \begin{align*}
    B_{i + 1, \geq j'}       %
    \tago{\subseteq}                      %
    \spn{B_{i, \geq j'}}           %
    \subseteq                      %
    \spn{B_{i, \geq j'} + e}       %
    =                              %
    \spn{B_{i, \geq j'}' + d}
    \tago{=}                      %
    \spn{B_{i, \geq j'}'},          %
  \end{align*}
  and taking the span of both sides gives the desired subset
  inequality. Here \tagr is by the descending order, \tagr is by
  monotonicity, \tagr is because
  $B_{i, \geq j'}' + d = B_{i, \geq j'} + e$, and \tagr is because $d$
  is spanned by $B_{i, \geq j'}' = B_{i, \geq j'} + e - d$.

  Next we show that greedy insertion maintains the push-relabel
  \IMPInvariants. More directly, we will show that greedy insertion
  simulates a legal sequence of basic operations; as mentioned, \FM
  has already proven that basic operations maintain the invariants.

  Given $e$, suppose we repeatedly try to push or relabel $e$ until we
  either (a) execute a push, or (b) increase $\level{e}$ to
  $\Height$. In event (a), of all possible choices of base $\basei$ in
  which to exchange $e$ at a fixed level, we specifically select the
  base $\basei$ with the smallest index $i$. This describes a valid
  sequence of basic operations, hence would preserve
  \IMPInvariants. We will show that greedy insertion simulates this
  process.

  Fix an uncovered element $e$ with $\level{e} = j$.  First, we claim
  that an element $e$ can be pushed into a base $\basei$ (at that
  level) iff $e$ is not spanned by $\basei_{\geq j}$. Indeed, if
  $\basei_{\geq j}$ spans $e$, then all elements $d \in \basei$ that
  could be exchanged for $e$ are in $\basei_{\geq j}$, and in
  particular, not in $\basei_{j - 1}$ as required for a
  push. Conversely, if $\basei_{\geq j}$ does not span $e$, then the
  unique circuit of $\basei + e$ must contain an element $d$ with
  $d \in \basei_{< j}$. Moreover, by \cref{imp-span},
  $\basei_{\geq j-1}$ spans $e$, so this circuit is also the unique
  circuit of $\basei_{\geq j-1} + e$. This implies that
  $d \in \basei_j$, and can be exchanged for $e$.

  Second, we observe that for all $i$, if
  $e \notin \spn{\basei_{\geq j}}$, then
  $e \notin \spn{\basei[i']_{\geq j}}$ for all $i' \geq i$. This
  follows from the decreasing order of bases which implies that the
  sets $\spn{\basei_{\geq j}}$ forms a nested, descending sequence of
  sets. This observation implies the following two points. First, if
  $e \in \spn{\lastB_{\geq j}}$, then $e$ is spanned by all
  $\basei_{\geq j}$, and we can safely increase $\level{e}$. Second,
  if $e \notin \spn{\basei_{\geq j}}$ for \emph{some} $\basei$, then
  we can binary search for the base $\basei$ with smallest index $i$
  such that $e \notin \spn{\basei_{\geq j}}$.

  Putting everything together, recall that we want to argue that
  greedy insertion simulates a push/relabel process that repeatedly
  relabels $e$ until either $\level{e} = \Height$ or we can exchange
  $e$ into a base $\basei$ for an element
  $d \in \basei_{\level{e} - 1}$, at which point it makes the exchange
  into the first such $\basei$. As observed above, such an exchange is
  possible iff $e \notin \spn{\basei_{\geq \level{e}}}$. Moreover, as
  observed above, the latter is possible iff
  $e \notin \spn{\lastB_{\geq \level{e}}}$.  Since the sets
  $\lastB_{\geq j}$ forms a nested, decreasing sequence of sets in
  $j$, and $\spn{\cdots}$ is a monotonically increasing set function,
  we can binary search for the first (smallest) index $j$ such that
  $e \notin \lastB_{\geq j}$. This index $j$ is exactly the level that
  the simulated push-relabel process would have eventually set
  $\level{e}$ to.  Assuming $j < h$, the simulated push-relabel
  process would then identify the first $\basei$ into which we can
  exchange $e$. By the observations above, this base $\basei$ is the
  same as that identified via binary search in step
  \ref{imp-greedy-2}.
\end{proof}

\Reflemma{impr-greedy-update} establishes the correctness of the
algorithm via the \IMPInvariants. To complete the proof of \cref{impr}
it remains to prove the running time bound.

Each instance where we greedily insert an uncovered element $e$ can be
charged to either (a) setting $\level{e} = \Height$, or (b) increasing
the size of $B_{i, >j}$ for some $i \in [k]$ and $0 \leq j < \Height$.
Each element has its level set to $\Height$ once, so there are $n$
insertions of type (a). Each insertion of type (a) takes one
independence query. To bound the number of insertions of type (b), we
observe that for a fixed level $j > 0$, the sets $B_{i, > j}$ across
$i$ form a feasible solution to \refequation{max-imp}, hence have
total size at most $\OPT$. Since there are $\Height$ levels, we have
at most $\Height \OPT$ insertions of type (b).  So to recap, we have
at most $n$ greedy insertions of the first type and
$\bigO{\OPT \Height}$ of the second type.

The first type of greedy insertion takes one oracle call. Consider the
second type.  With binary search, the first search in step
\ref{imp-greedy-1} takes $\bigO{\log \Height}$ probes. Better yet, by
standard doubling tricks, we can adjust the binary search so that the
first search also takes at most $\bigO{1 + \ell}$ probes where $\ell$
is the number of levels the element moves forward.  We will be able to
charge these off to increasing the ranks of at least $\ell$ sets
$\basei{> j}$.  The search in step \ref{imp-greedy-2} takes
$\bigO{\log k}$ probes. In both cases, each probe takes one
independence query.

When executing an exchange, we also need to identify an element $d$ to
remove quickly. To this end, we can maintain a balanced binary tree
over $\basei_{j}$ in insertion order, and use binary search to quickly
identify the last possible choice of $d$. (This is the first element
$d$ such that all the elements in $\basei_{j-1}$ before $d$, along
with the elements in $\basei_{\geq j}$, do not span $e$). This takes
$\log{\therank}$ independence queries.
We note that step \ref{imp-greedy-2} is only invoked for one of the
$\OPT \Height$ queries of type (b). This gives the total running time.

This completes the presentation of the faster matroid push-relabel
algorithm for integer capacities.

\subsection{Putting it all together}

By combining the running time of \reftheorem{impr} with the required
heights per \cref{imp-exact-analysis,imp-apx-analysis}, we obtain the
following running times for exact and approximate $k$-fold matroid
union with integer capacities.

\begin{theorem}
  \labeltheorem{mp-exact} \labeltheorem{matroid-union} For integer
  capacities, a maximum $k$-fold matroid union and a dual solution can
  be computed in $\MatroidUnionTime$
  independence queries.
\end{theorem}

\begin{theorem}
  \label{apx-matroid-union}
  \labeltheorem{mp-apx} For integer capacities, a
  $\epsless$-approximately maximum $k$-fold matroid union, and a
  $\epsmore$-approximately minimum dual solution, can be computed in
  time bounded by $\ApxMatroidUnionTime$ independence queries.
\end{theorem}


\section{Base Packing and Covering}

\labelsection{base-packing-covering} %

We now turn to the problems of packing and covering a matroid in
bases. The problems were briefly introduced in \cref{results}
and we now describe them in greater detail.

In the base packing problem, given an integer $k$, the goal is to
compute $k$ bases $\Bs$ such that $\x{e} \leq \capacity{e}$ for all
$e \in \groundset$.  We say that $\Bs$ is a \emph{packing} when
$\x{e} \leq \capacity{e}$ for all $e$.
\citet{Edmonds1965a} proved
there is a packing $\Bs$ of $k$ bases iff
\begin{align*}
  \capacity{\bar{S}} \geq k \parof{\therank - \rank{S}}
\end{align*}
for all sets $S$.

In the base covering problem, the goal is to compute $k$ bases $\Bs$
such that $\x{e} \geq \capacity{e}$ for all $e \in \groundset$. Such a
set of $\Bs$ is called a \emph{covering}.  \citet{Edmonds1965c} proved
there is a covering $\Bs$ of $k$ bases iff
\begin{align*}
  k \rank{S} \geq \capacity{S}
\end{align*}
for all sets $S \subseteq \groundset$.

Both of the dual characterizations above, for base packing and for
base covering, can be obtained via the dual characterization for
$k$-fold matroid union of \cite{NashWilliams1967} that was presented in
\refsection{impr}.

As mentioned there, exact base packing and covering reduces directly
to $k$-fold matroid union. For example, for packing, there is a base
packing of $k$ bases iff there are $k$ bases whose union has size
$k \therank$.  For covering, there is a base covering of $k$ bases iff
there are $k$ bases whose union has size $n$.  The dual solutions
given by the $k$-fold matroid union algorithm also provide a
certificate of infeasibility for base packing or covering, when no
packing or covering is found.

One might expect the approximation algorithms for $k$-fold matroid
union to also be an approximation algorithms for packing and covering,
but this is not the case. Consider, for example, uncapacitated packing
of $k$ bases. The approximation algorithm for $k$-fold matroid union
will output $k$ bases whose union has $\epsless$-times the maximum
size of any union. If there exists $k$ disjoint bases, then in
particular the union has size at least $\epsless k \therank$ total
elements. However this is not the same as $\epsless k$ disjoint
bases. A union of size $\epsless k \therank$ neither confirms that
there are at least $\epsless k$ disjoint bases, nor denies that there
exist $k$ disjoint bases.

Similar disparities arise for capacitated packing, and capacitated and
uncapacitated covering. For all of these problems, we give a slightly
stronger analysis to obtain the desired form of approximation. The
main difference here is that height necessary to obtain
$\epspm$-approximations increases by a logarithmic
factor. Consequently all the running times for approximating packing
and covering are a logarithmic greater than for approximating $k$-fold
matroid union.

\subsection{Packing}

We first consider approximations for uncapacitated base packing, for
which \cref{ubp} claimed a running time of
$\bigO{n \log{1/\eps} + n \log{k \therank / \eps}}$ independence
queries. Below we prove that height $\bigO{\log{n} / \eps}$ height
implies a $\epsless$-approximation. The running time in \cref{ubp}
then follows from running the uncapacitated matroid push-relabel
algorithm for height $\bigO{\log{n} / \eps}$, by \cref{impr}.

\begin{lemma}
  \label{apx-uncapacitated-packing}\label{apx-up-analysis}
  Let $\eps \in (0,1)$ and $k \in \naturalnumbers$.  Let $\Bs$ be a
  family of $k$ bases and $\deflevels$ an assignment of levels
  satisfying the matroid push-relabel \IMPInvariants with height
  $\Height = \bigO{\log{n} / \eps}$. Then either (a) the bases $\Bs$
  forms a feasible packing of $k$ bases, or (b)
  $(\matroid, \capacity)$ has no feasible packing of $\epsmore k$
  bases.
\end{lemma}

\begin{proof}
  Suppose that $\Bs$ is not a proper packing.  By \cref{imp-disjoint},
  this implies that $\level_0 \neq \emptyset$.  Now,
  $\capacity{\level_{\leq j}}$ is nondecreasing in
  $j \in \setof{0, \dots, \Height}$, bounded below by $1$ for $j = 0$
  (as noted above) and at most $n$ for $j = \Height$. Consequently
  there must be a level $j \in \setof{1,\dots,\Height-1}$ such that
  $\capacity{\level_j} < \parof{\eps / (1+\eps)}
  \capacity{\level_{\leq j}}$. Fix $j$ as such.
  By choice of $j$ we have
  \begin{align*}
    \capacity{\level_{\leq j}} + \epsmore k \rank{\groundset_{> j}}
    <
    \epsmore \parof{\capacity{\level_{\leq j}} +  k \rank{\level_{>
    j}}} - \epsmore \capacity{\level_j}
    \labelthisequation{apx-packing-analysis-1}
  \end{align*}

  By \cref{imp-covered}, all elements in $\level_{\leq j}$ are
  covered, hence
  \begin{align*}
    \capacity{\level_{\leq j}} \leq \epsmore \capacity{\level_{< j}}
    \leq \epsmore \x{\level_{\leq j}}.
  \end{align*}
  By \cref{imp-span}, each $\basei{\geq j}$ covers $\groundset_{> j}$,
  and none of the elements in $\groundset_{\geq j}$ are overpacked,
  hence
  \begin{align*}
    k \rank{\level_{> j}} \leq \sum_{i=1}^k \sizeof{\basei_{\geq j}}
    =                        %
    \x{\groundset_{\geq j}}. %
  \end{align*}
  Altogether we have
  \begin{align*}
    \capacity{\level_{\leq j}} + \epsmore k \rank{\groundset_{> j}}
    \leq                        %
    \epsmore \x{\groundset}
    \leq                        %
    \epsmore k \therank.
  \end{align*}
  Rearranging, we have
  \begin{align*}
    \capacity{\level_{\leq j}}
    \leq                        %
    \epsmore k \parof{\therank - \rank{\level_{> j}}},
  \end{align*}
  hence $\level_{> j}$ certifies that the strength is at most
  $\epsmore k$.
\end{proof}

Applying the matroid push-relabel algorithm (\reftheorem{impr}) with
the height parameter $\Height = \bigO{\log{n} / \eps}$, per
\cref{apx-up-analysis}, gives the following $\epsless$-approximation
algorithm for base packing.

\begin{theorem}
  \label{apx-pack-bases}\label{apx-base-packing}
  \label{ubp}\label{cbp}%
  An $\epsless$-approximate base packing, or an $\epsmore$-approximate
  certificate of infeasibility, can be computed in running time
  bounded by
  \begin{math}
    \ApxBasePackingTime
  \end{math}
  independence queries.
\end{theorem}

\subsection{Covering}

We now move on to covering problems, starting with uncapacitated
covering. Here we show that $\bigO{\log{\therank} / \eps}$ height
suffices to obtain an $\epsmore$-approximation; note that this bound
is (slightly) better than for approximating uncapacitated packing
above.

\begin{lemma}
  \label{apx-uncapacitated-covering}\label{apx-uc-analysis}
  Let $\defmatroid$ be a matroid and $k \in \naturalnumbers$. Let
  $\Bs \in \bases$ be a collection of $k$ bases and $\deflevels$ a set
  of levels that satisfy the push-relabel \IMPInvariants with height
  $\Height = \bigO{\log{\therank} / \eps}$. Then either:
  \begin{enumerate}
  \item $\Bs$ is a covering, or
  \item For some index $j$, $\level_{> j}$ certifies that more than
    $\epsless k$ bases are required in any covering.
  \end{enumerate}
\end{lemma}

\begin{proof}
  Suppose $\Bs$ is not a covering. Then there is at least one
  uncovered element; moreover, any uncovered element is in
  $\groundset_{\geq \Height}$ by \cref{imp-uncovered}. As a function of
  $j \in \setof{0,\dots,\Height}$, $\rank{\level_{\geq j}}$ is
  nonincreasing, bounded above by $\therank$ at $j = 0$, and bounded
  below by $1$ at $j = \Height$ because $\level_{\geq \Height}$ is
  nonempty. Since $\Height$ is at least
  $\bigO{\log{\therank} / \eps}$, there must be a level
  $j \in \setof{0,\dots,\Height}$ such that
  $\rank{\level_{> j}} \geq \epsless \rank{\level_{\geq j}}$. Fix $j$
  as such.  We have
  \begin{align*}
    \capacity{\groundset}
    >                           %
    \sum_{e \in \groundset} \min{\capacity{e}, \x{e}} \labelthisequation{apx-covering-analysis-1}
  \end{align*}
  because $\Bs$ is not a covering. Since elements in
  $\groundset_{\geq j}$ are not overpacked, we have
  \begin{align*}
    \sum_{e \in \groundset_{\geq j}} \min{\capacity{e}, \x{e}}
    =
    \sum_{e \in \groundset_{\geq j}} \x{e}
    =                           %
    \sum_{i=1}^k \sizeof{\basei{\geq j}}.
  \end{align*}
  By \cref{imp-span} each $\basei{\geq j}$ spans $\groundset_{> j}$
  hence
  \begin{align*}
    \sum_{i=1}^k \sizeof{\basei{\geq j}} \geq k \rank{\level_{> j}}.
  \end{align*}
  On the other hand, since all elements in $\groundset_{< j}$ are
  covered, we have
  \begin{align*}
    \sum_{e \in \groundset_{< j}} \min{\capacity{e}, \x{e}}
    = \capacity{\groundset_{< j}}.
  \end{align*}
  Plugging back into \refequation{apx-covering-analysis-1} we now have
  \begin{align*}
    \capacity{\groundset} > k \rank{\level_{> j}} +
    \capacity{\groundset_{< j}}.
  \end{align*}
  Finally, since $\rank{\level_{> j}} \geq \epsless \rank{\level_{\geq
      j}}$ by choice of $j$, we obtain
  \begin{align*}
    \capacity{\groundset} >\epsless k \rank{\level_{\geq j}} + \sizeof{\level_{< j}}.
  \end{align*}
  Rearranging we have that
  $\epsless k \rank{\level_{\geq j}} < \sizeof{\level_{\geq j}}$,
  which by the dual characterization above implies
  that more than $\epsless k$ bases are required in any covering.
\end{proof}

Plugging in $\Height = \bigO{\log{\therank} / \eps}$ to the matroid
push-relabel algorithm (\reftheorem{impr}), we obtain the following
approximation algorithm for base covering.  (When plugging into
\reftheorem{impr}, we note that $\opt = n$, and $n \leq kr$.)

\begin{theorem}
  \label{apx-base-cover}\label{apx-base-covering}
  \label{ubc}\label{cbc}
  For integer capacities, one can compute either a covering of
  $\epsmore k$ bases or a certificate of infeasibility for any
  covering of $\epsless k$ bases in time bounded by
  $\ApxBaseCoveringTime$ independence queries.
\end{theorem}

\begingroup

\renewcommand{\theElements}{E}%
\renewcommand{\groundset}{E}%
\newcommand{\trees}{\mathcal{T}}%

\section{Graphic Matroid Push-Relabel}

\labelsection{spanning-trees}   %

In this section, we consider the matroid push-relabel framework for
the special case of the graphical matroid; i.e., forests of an
undirected graphs. This leads to algorithms for the several graph
problems mentioned in \cref{results}

\begin{theorem}
  \label{uncapacitated-graphic-push-relabel} \label{ugpr} Given a
  graph with integer capacities and integers $k$ and $h$, in
  $\bigO{m \ack{n} + \opt h \parof{\log{n} + \log{k} \ack{n}}}$ time,
  one can compute a sequence of $k$ spanning trees $T_1,\dots,T_k$ and
  levels $\deflevels$ satisfying the push-label \IMPInvariants with
  height $h$.
\end{theorem}

\begin{proof}
  We implement the matroid push-relabel algorithm with the following
  data structures.  We maintain, for each spanning tree $T_i$, a
  link-cut tree \cite{SleatorTarjan1983}, with edges labeled by their
  levels.  Given an edge $e$, a tree $T_i$, we can query for the
  biggest level $j$ such that $T_{i, \geq j}$ spans $e$ by querying
  for the minimum level edge on the cycle induced by $e$. This also
  allows us to retrieve an edge $d$ to exchange out in $\bigO{\log n}$
  time. We make a total of $\bigO{\opt h}$ such exchanges.

  We also maintain, for each base $\basei$ and each level
  $j \in \setof{1,\dots,\Height}$, a disjoint union data structure
  representing the connected components of $\basei{\geq j}$. This
  allows us to query if an edge $e$ is spanned by a forest
  $\basei_{\geq j}$ in $\ack{n}$ time.  We make $\bigO{\opt h}$
  total insertions into thee disjoint union data structures over all
  $i$ and $j$.  We make at most $\bigO{m + \opt h \log{k}}$ such
  queries.
\end{proof}

\Cref{uncapacitated-graphic-push-relabel}, combined with the
optimality conditions given by
\cref{imp-exact-analysis,imp-apx-analysis,apx-up-analysis,apx-uc-analysis}
for $k$-fold matroid union, base packing, and base covering, both
exact and approximate, give the follow running times for the graphic
matroid.

We start with the $k$-fold matroid union problem. For graphs it is
more natural to state this as computing a packing of forests
$F_1,\dots,F_k$ of maximum total capacity.\footnote{As with trees, a
  set of forests is a packing if no element $e$ appears in more than
  $\capacity{e}$ of the forests.}  We have the following exact and
approximate running times.

\begin{theorem}
  \label{max-forests}
  A maximum capacity packing of $k$ forests, and the dual minimization
  problem, can be solved in
  $\bigO{m \ack{n} + n \opt \parof{\log{n} + \log{k} \ack{n}}}$ time,
  where $\opt$ denotes the optimum size.
\end{theorem}

\begin{theorem}
  \label{apx-max-forests}
  An $\epsless$-approximately maximum capacity packing of $k$ forests,
  and an $\epsmore$-approximately minimum dual solution, can be
  computed in
  $\ApxForestUnionTime$ time.
\end{theorem}

For packing and covering spanning trees, in addition to the exact
algorithms implied by \cref{max-forests}, we have the following
approximation algorithms.
\begin{theorem}
  For integer-capacitated graphs, there is an algorithm that, in
  \begin{math}
    \ApxTreePackingTime
  \end{math}
  time, outputs either a packing of $k$ spanning trees, or a
  certificate that the network strength is less than $\epsmore
  k$.
\end{theorem}

\begin{theorem}
  In integer-capacitated graphs, there is an algorithm that, in
  $\ApxTreeCoveringTime$ time, outputs either a covering by $k$
  spanning trees, or a certificate that the strength is less than
  $\epsmore k$.
\end{theorem}


\endgroup

\section{Minimum Cost Reinforcement}

\labelsection{reinforcement}    %

\newcommand{\minB}{B_0}
\providecommand{\cost}{\fparnew{c}}%

In this section we consider the minimum cost reinforcement problem,
introduced in \cref{results}. We primarily discuss the more
general matroid setting; the graphic setting follows as a special
case.

Let $\defmatroid$ be a matroid with $n$ elements and rank $\therank$,
let $\capacity: \groundset \to \nnintegers$ be a set of integer
capacities, let $\cost: \groundset \to \preals$ be a set of
real-valued costs, and let $k \in \naturalnumbers$. The cost
$\cost{e}$ represents the cost of augmenting $\capacity{e}$ by
$1$. The goal is to compute the minimum cost augmentation of
$\capacity$ to obtain strength $k$.

Our argument is closely tied to \citets{Cunningham1985c}{algorithm},
which was the first strongly polynomial time algorithm for this
problem. \citeauthor{Cunningham1985c} focused on network strength for
undirected graphs and here we describe a straightforward
generalization of his algorithm to matroids.
\begin{enumerate}
\item Compute a maximum point $y \in k\ip$ subject to
  $y \leq \capacity$, via $k$-fold matroid union.
\item Greedily extend $y$ to an integral point $y + z$ in $k \ip$,
  where $z \in \nnintegers^{\groundset}$ is computed as
  follows. Initially, we set $z(e) = 0$ for all elements $e$. Then for
  each element $e$ in increasing order of costs, set $z(e)$ as large
  as possible subject to $y + z \in k \ip$. The maximum value for
  $z(e)$ can be obtained by binary search, where each probe invokes a
  $k$-fold matroid union algorithm to see if the candidate value for $z(e)$
  is feasible.
\end{enumerate}
We refer to \cite{Cunningham1985c} for the full justification of this
algorithm. Note that the algorithm requires many calls to a matroid
partition algorithm; first to compute the initial point $y$, and then
logarithmically many times for every element $e$ to obtain the right
$z(e)$.

In what is perhaps a surprising coincidence, the $k$-fold matroid
union algorithm developed in \refsection{impr} actually solves the
reinforcement problem in one shot (so to speak). There is just one
minor adjustment: the initial bases (which was allowed to be arbitrary
in \refsection{impr}) must all be set to the minimum cost base which
we denote $\minB$. Recall that the push-relabel algorithm applied to a
set of integer capacities $\capacity$ and a parameter $k$ produces a
set of $k$ bases $\Bs$ that maximizes
\begin{align*}
  \sum_{e \in \groundset} \min{\capacity{e}, \x{e}},
\end{align*}
where $\x{e}$ denotes the number of bases $\basei$ containing $e$.
The vector $y$ defined by $y(e) = \min{\capacity{e}, \x{e}}$ (for all
$e$) fulfills step 1 of \citeauthor{Cunningham1985c}'s algorithm. The
key claim, proven below, is that if the bases in the $k$-fold matroid
union algorithm are all initially set to the minimum cost base
$\minB$, then at the end of the algorithm, $\x$ describes $y + z$ for
an optimum reinforcement solution $z$. Thus $z$ can be read off
directly from $x$ and $y$.  This gives an overall running time that is
exactly the same as for $k$-fold matroid union.  Here we have two
speed-ups compared to \cite{Cunningham1985c} -- one from a faster
$k$-fold matroid union algorithm, and the second from omitting the
second stage altogether. The following lemma formalizes the key claim.

\begin{lemma}
  Let $\minB$ be the minimum cost base w/r/t $\cost$, and consider the
  exact push-relabel $k$-fold matroid union algorithm adjusted so that the
  initial bases are all set to $\minB$. Let $\Bs$ be the $\numB$ bases
  output by the push-relabel matroid-partition algorithm, and let
  $x \in \nnintegers^{\groundset}$ be the vector where $x(e)$ is the
  number of bases $\basei$ containing $e$ for each $e \in \groundset$.
  Define $z \in \nnreals^{\groundset}$ by
  \begin{align*}
    z(e) = \min{0, \x{e} - \capacity{e}} \text{ for } e \in \groundset.
  \end{align*}
  Then $z$ is a minimum cost reinforcement.
\end{lemma}

\begin{proof}
  In addition to $\Bs$ and $\x$ as described above, let $\level$ be
  the set of levels produced by the capacitated $k$-fold matroid union
  algorithm. Then $\Bs$, $\numCovered$, and $\level$ satisfy
  \IMPInvariants from \cref{imp-framework}.  Let
  $C = \setof{e \where \x{e} > \capacity{e}}$. We have $\level{e} = 0$
  for all $e \in C$ by \cref{imp-level-0}.
  \begin{claims}
  \item $C \subseteq \minB$.
  \end{claims}
  To this end, observe that initially we have $\x{e} > \capacity{e}$
  only for elements in the initial base, $e \in \minB$. Thereafter, a
  coordinate $\x{e}$ is only increased if $\x{e} < \capacity{e}$,
  and never exceeding $\capacity{e}$.

  \begin{claims}
  \item \labelclaim{mcr-claim-2} For any element $e \in \minB$, and
    any base $\basei \in \setof{\Bs}$, we have either $e \in \basei$,
    or $e \in \spn{\basei \setminus C}$.
  \end{claims}
  Fix any element $e \in \minB$ and a base $\basei$ from $\Bs$. If
  $e \notin \basei$, then it was exchanged out by an element $d$ with
  $\level{d} = 1$, such that $\basei_{\geq 1} = \basei_{> 0}$ spans
  $e$. That is, if $e \notin \basei$, then $e \in \spn{\basei_{>
      0}}$. Since $C \subseteq \level_0$, we have
  $\basei_{> 0} \subseteq \basei \setminus C$, hence
  $e \in \spn{\basei \setminus C}$.

  Now, let $y$ be the pointwise minimum of $\numCovered$ and
  $\capacity$; $y$ is maximum in $k \ip$ subject to
  $y \leq \capacity$.  Recall that a minimum cost base can be produced
  by a greedy algorithm adding feasible elements in nondecreasing
  order of cost.  Number the elements
  $\groundset = \setof{e_1,e_2,\dots}$ in nondecreasing order of cost,
  breaking ties so that a greedy algorithm processing elements in this
  order produces $\minB$. In Cunningham's greedy augmentation
  algorithm, a minimum cost reinforcement $z$ is obtained by
  processing the $e_i$'s in order, taking setting $z(e_i)$ to the
  maximum quantity subject to $y + z \in k \ip$. (This equals the
  minimum quantity subject to $\parof{y+z}/k$ spanning $e_i$ in the
  independent set polytope $\ip$, and in Cunningham's algorithm it is
  identified via binary search and a call to $k$-fold matroid union
  for each probe).  \citeauthor{Cunningham1985c} has already shown
  that this algorithm produces an optimum solution. Therefore it
  suffices to prove the following claim.
  \begin{claims}
  \item Cunningham's greedy augmentation selects
    $z(e) = \max{\x{e} - \capacity{e}, 0}$ for all $e \in \groundset$.
  \end{claims}

  We analyze each element in the greedy order.  Consider the $i$th
  iteration (where $i \in [n]$), in which the greedy algorithm
  processes $e_i$. We assume by induction that $(y + z)/k$ spans
  $\setof{e_1,\dots,e_{i-1}}$ (in $\ip$). (The base case, where
  $i = 1$, holds vacuously.) If $e_i \notin \minB$, then
  $e_i \in \spn{\setof{e_j \in \minB \where j < i}}$, so $e_i$ is
  spanned by $(y+z)/k$.

  Now suppose $e_i \in \minB$.  We claim that the greedy augmentation
  algorithm sets $z(e_i) = \x{e_i} - \capacity{e_i}$. To see this, let
  $z'$ be the vector obtained from $z$ by setting
  $z'(e_i) = \x{e_i} - \capacity{e_i}$. We know that $y + z'$ is
  feasible because $y + z' \leq \x$ and $\x \in k \ip$.  To show that
  $(y + z')/k$ spans $e_i$, we first observe that one can pack into
  $y + z'$ the $k$ independent sets
  \begin{align*}
    \basej' \defeq \basej
    \setminus \parof{C - e_i} \text{ for } j = 1,\dots,k.
  \end{align*}
  We claim that $\basej'$ spans $e_i$ for each $j$ which shows that
  $(y+z')/ k$ spans $e_i$.  We have two cases. In the first case, if
  $e_i \in \basej'$, then of course $e_i \in \spn{\basej'}$. In the
  second case, if $e_i \notin \basej'$, then $e_i \notin \basej$. By
  \refclaim{mcr-claim-2}, $e_i \in \spn{\basej \setminus C}$. Since
  $\basej' \subseteq \basej \setminus C$, $e_i \in \basej'$.

  This shows that the Cunningham's greedy augmentation algorithm takes
  $z(e_i) = \x{e_i} - \capacity{e_i}$. This establishes the claim, and
  completes the proof.
\end{proof}


\section{Approximations for problems with general capacities}

\labelsection{apx-capacities}

\newcommand{\rcap}{\fparnew{\smash{\tilde{u}}\vphantom{u}}}%
\newcommand{\varrank}{q}%
\newcommand{\constA}{c_1}%
\newcommand{\constB}{c_2}%
\newcommand{\constC}{c_3}%
\newcommand{\constD}{c_4}%

This section develops fast approximation algorithms for matroid
problems for general capacities. All of the algorithms in the section
is based on using \emph{randomized rounding} to reduce problems with
real-valued capacities and a real-valued parameter $k$ to problems
with integer capacities and an integer parameter $k$ on the order of
$\ln{n} / \eps^2$, with high probability.  We then apply the
approximate push-relabel algorithms developed in prior sections which
are particularly well suited to the reduced setting.

One cost of this convenience is that we will no longer obtain primal
solutions for the original input. However we will still be able to
obtain dual solutions which at least provide a certificate for one
side of the corresponding decision problem.

\subsection{Randomized rounding of real-valued capacities}

\labelsection{sparsification}

\newcommand{\comp}[1]{\smash{\bar{#1}}\vphantom{#1}}

\label{sparsification-setup}

Let $\defmatroid$ be a matroid with $n$ vertices and rank $\therank$,
and let $\capacity: \groundset \to \preals$ be a set of
capacities. Let $k > 0$ be a parameter specified by the context.  Let
$\tau > 0$ also be a given parameter with
$\tau \leq c k \eps^2 / \ln{n}$ for a sufficiently small constant
$c$. Decreasing $c$ as needed, we may assume that $k / \tau$ is an
integer without loss of generality.  Let
$\rcap \in \nnintegers^{\groundset}$ be the a randomized set of
integral capacities by randomly rounding $u / \tau$ to an integral
vector. That is, for each element $e \in \groundset$, we independently
set
\begin{align*}
  \rcap{e} =
  \begin{cases}
    \roundup{\capacity{e} / \tau} & \text{with probability }
                                    \capacity{e} / \tau -
                                    \rounddown{\capacity{e} / \tau}, \\
    \rounddown{\capacity{e} / \tau} & \text{with (remaining) probability
                                      }
                                      \roundup{\capacity{e} / \tau} -
                                      \capacity{e} / \tau.
  \end{cases}
\end{align*}

The scaled down capacitated matroid $(\matroid,\capacity)$ have some
immediately appealing properties. First we have
$\evof{\rcap{e}} = \capacity{e} / \tau$ for all $e \in \groundset$.
By linearity of expectation we also have
$\evof{\rcap{S}} = \capacity{S} / \tau$ for all sets $S$. Second,
$\parof{\matroid, \rcap}$ has an integer capacities, that are (in
expectation) a $\apxO{k}$-factor smaller than $(\matroid,
\capacity)$. We are interested in applying this randomized rounding
for problems such as maximizing the total capacity covered by a
packing of $k$ independent sets, packing $k$ bases, or covering by $k$
bases, and these problem-specific values of $k$ are used for the value
of $k$ in randomly rounding to $\rcap$. Therefor it is helpful that
$k / \tau$ is an integer for the corresponding scaled down problems
over $(\matroid, \rcap)$.  Additionally, for all these problems, $k$
is a natural upper bound or near-upper bound on the capacities, hence
$\rcap$ will have relatively small capacities bounded above by
$\bigO{\ln{n} / \eps^2}$ in expectation (and with high probability).

We would like to show that the $\parof{\matroid,\rcap}$ is (with high
probability) a good representative sample of $(\matroid, \capacity)$
for these problems. Now, while $\rcap$ reflects $\capacity / \tau$ in
expectation, in general the values $\rcap{e}$ for $e \in \groundset$
and $\rcap{S}$ for $S \subseteq \groundset$ are too numerous to assume
they are all concentrated at their expectation in expectation. (In
fact, a value $\rcap{S}$ will never be close to its expectation,
multiplicatively speaking, when $\capacity{S}$ is significantly
smaller than $\tau$.) Nonetheless we have the following theorem which
leverages the dual characterizations of these problems to show that
$(\matroid, \rcap)$ is a good (problem-specific) representation of
$(\matroid,\capacity)$ with high probability.

\begin{theorem}
  \label{general-sparsification}
  Given the setup described above, the following all hold with high
  probability.
  \begin{mathresults}
  \item \label{randomly-round-union}\label{RRU} Letting $M$ denote the maximum total capacity of any
    fractional packing of $k$ independent sets in
    $(\matroid,\capacity)$, and $\tilde{M}$ denote the maximum total
    capacity of any fractional packing of $k/\tau$ independent sets in
    $(\matroid, \rcap)$, we have
    \begin{math}
      \absvof{\tilde{M} - M/\tau} \leq \eps M/ \tau.
    \end{math}
  \item \label{randomly-round-packing->}\label{RRP>} If $(\matroid,\capacity)$ can
    fractionally pack $k$ bases, then $(\matroid, \rcap)$ can pack at
    least $\epsless k / \tau$ bases.
  \item \label{randomly-round-packing-<}\label{RRP<} If $(\matroid,\capacity)$
    cannot fractionally pack $k$ bases, then $(\matroid, \rcap)$
    cannot pack $\epsmore k / \tau$ bases.
  \item \label{randomly-round-covering-<}\label{RRC<} If $(\matroid, \capacity)$
    can be covered by $k$ bases, then $(\matroid, \rcap)$ can be
    covered by $\epsmore k / \tau$ bases.
  \item \label{randomly-round-covering->}\label{RRC>} If $(\matroid, \capacity)$
    cannot be covered by $k$ bases, then $(\matroid, \rcap)$ cannot be
    covered by $\epsless k / \tau$ bases.
  \end{mathresults}
\end{theorem}

\begin{proof}
  Recall that a set $S \subseteq \groundset$ is \emph{closed} if
  $S = \spn{S}$. We claim that with high probability, we have
  \begin{align*}
    \absvof{\capacity{S} - \tau \rcap{S}} \leq \frac{\eps}{2}
    \parof{\capacity{S} + k \rank{S}}
    \text{ and }
    \absvof{\capacity{\comp{S}} - \tau \rcap{\comp{S}}} \leq \frac{\eps}{2}
    \parof{\capacity{\comp{S}} + k \rank{S}}
    \labelthisequation{ranked-concentration}
  \end{align*}
  for all closed sets $S$ (simultaneously).

  For ease of notation, call a closed set $S$ \emph{bad} if $\rcap{S}$
  or $\rcap{\comp{S}}$ violates the inequalities above. We want to
  prove that there are no bad closed sets with high probability.

  First, fix a closed set $S$ with $\rank{S} = \varrank$. Consider the
  first (leftmost) of the inequalities we seek. By standard Chernoff
  inequalities, we have
  \begin{align*}
    \probof{\absvof{\rcap{S} - \capacity{S} / \tau} > \frac{\eps}{2}
    \parof{\capacity{S} / \tau + k q / \tau}}
    \leq
    2 e^{- \eps^2 k q / 4 \tau}
    =                           %
    2 n^{- q / 4 c}.
  \end{align*}
  Likewise the second inequality (for $\rcap{\comp{S}}$) has probability of
  error is at most $2 n^{-q / 4 c}$. Taking the union bound,
  \begin{align*}
    \probof{S \text{ is bad}} \leq 4 n^{- q / 4 c} \leq n^{- \constA q}
  \end{align*}
  for a sufficiently large constant $\constA$.\footnote{By which we
    mean that $\constA$ can be made an arbitrarily large constant by
    making $c$ sufficiently small.}

  Now, fix $\varrank$. Each closed set $S$ is defined by any base of
  $S$, which consists of $\varrank$ elements. Therefore, there are at
  most $n^{\varrank}$ closed sets of rank $\varrank$. Taking the union
  bound over all sets $S$ of rank $\varrank$,
  \begin{align*}
    \probtext{any closed set of rank $\varrank$ is bad}
    \leq                        %
    2 n^{\varrank} n^{-\constA q}
    \leq
    n^{- \constB q}
  \end{align*}
  for a sufficiently large constant $\constB$.

  Finally, taking the union bound over all ranks
  $\varrank \in [\therank]$, we have
  \begin{align*}
    \probtext{any closed set is bad}
    \leq                        %
    \sum_{\varrank=1}^{\therank}
    n^{-\constB q}
    \leq                        %
    n^{-\constC q}
  \end{align*}
  for a sufficiently large constant $\constC$. This proves the claim.

  For the rest of the proof we assume the high probability event where
  the inequalities in \refequation{ranked-concentration} hold for all
  closed $S$. We will use these inequalities to prove each of
  \cref{RRU,RRP>,RRP<,RRC>,RRC<}.

  Consider first \cref{RRU}. Recall that $M$ equals the minimum of
  \begin{math}
    \capacity{\comp{S}} + k \rank{S}
  \end{math}
  over all sets $S$, and similarly for $M'$ except \wrt $\rcap$ and
  $k/\tau$.  Since replacing $S$ with its closure $\spn{S}$ can only
  decrease this quantity, it suffices to consider only the closed
  sets.

  For all closed sets $S$
  \begin{align*}
    \rcap{\comp{S}} + \frac{k}{\tau} \rank{S} %
    &\geq                        %
    \frac{1-\eps}{\tau} \capacity{\comp{S}} -
    \frac{\eps k}{\tau} \rank{S} + \frac{k}{\tau} \rank{S}
    \\
    &=
      \frac{1-\eps}{\tau} \parof{\capacity{\comp{S}} + k \rank{S}}
      \geq
      \frac{\epsless M}{\tau}
  \end{align*}
  for all closed sets $S$. Thus $M' \geq \epsless M / \tau$ with high
  probability.

  Next we upper bound $M'$. There exists a set $S$ be a closed set
  such that $M = \capacity{\comp{S} + k \rank{S}}$. We have
  \begin{align*}
    M' &\leq \rcap{\comp{S}} + \frac{k}{\tau} \rank{S}
         \leq                       %
         \frac{1 + \eps}{\tau} \capacity{\comp{S}} +
         \frac{\eps k}{\tau} \rank{S} + \frac{k}{\tau} \rank{S}
    \\
       &=
         \frac{1+\eps}{\tau} \parof{\capacity{\comp{S}} + k \rank{S}}
         \leq
         \frac{\epsmore M}{\tau},
  \end{align*}
  as desired. This proves \cref{RRU}.

  Consider now \cref{RRP<,RRP>}. By the matroid base packing theorem, the packing number of
  $(\matroid,\capacity)$ is at least $k$ iff for all sets $S \subseteq
  \groundset$,
  \begin{align*}
    \capacity{\comp{S}} \geq k \parof{\therank - \rank{S}}.
  \end{align*}
  Thus the packing number is exactly $k$ if the inequality holds for
  all $S$, and is tight for some set $S$ with $\rank{S} < \therank$.

  Now consider $(\matroid,\rcap)$. By \cref{general-sparsification},
  with high probability, we have
  \begin{align*}
    \rcap{\comp{S}}
    &\geq \frac{1-\eps/2}{\tau} \capacity{S} - \frac{\eps
      k}{2\tau} \rank{S}
      \geq                        %
      \frac{(1-\eps/2) k}{\tau}
      \parof{\therank - \rank{S}} - \frac{\eps k}{2\tau} \rank{S}
    \\
    &\geq                        %
      \frac{\epsless k}{\tau} \parof{\therank - \rank{S}},
  \end{align*}
  so the packing number is at least $(1-\eps) k / \tau$. This proves
  \cref{RRP>}.

  For the opposite direction in \cref{RRP<}, we know there exists a
  closed set $S$ with $\rank{S} < \therank$ and
  $\capacity{\comp{S}} = k \parof{\therank - \rank{S}}$. Note that
  $\capacity{\comp{S}} \geq k$. By the Chernoff inequality we have
  \begin{align*}
    \probof{\rcap{\comp{S}} \geq \epsmore \capacity{\comp{S}}/\tau}
    \leq                        %
    e^{-\eps^2 \capacity{\comp{S}}/3 \tau} %
    \leq                           %
    e^{-\eps^2 k / 3 \tau}
    =
    n^{-1 / 3 c}.
  \end{align*}
  (We point out that $1/3c$ represents an arbitrarily large constant.)
  Thus with high probability we have
  \begin{align*}
    \rcap{\comp{S}} \leq \epsmore \capacity{\comp{S}} / \tau
    =
    \frac{\epsmore k}{\tau} \parof{\therank - \rank{S}}
  \end{align*}
  and so the packing number of $(\matroid, \rcap)$ is at most
  $\epsmore k / \tau$. This proves \cref{RRP<}.

  Lastly we prove \cref{RRC<,RRC>}. Recall that $(\matroid,\capacity)$
  can be fractionally covered by $k$ bases iff for all sets
  $S \subseteq \groundset$,
  \begin{math}
    \capacity{S} \leq k \rank{S}.
  \end{math}
  Since the capacities are nonnegative, it suffices to verify the
  inequality $\capacity{S} \leq k \rank{S}$ for all closed sets.

  Suppose $(\matroid,\capacity)$ can be fractionally covered by $k$
  bases.  We have
  \begin{align*}
    \rcap{S} \leq \parof{1 + \frac{\eps}{2}} \frac{\capacity{S}}{\tau}
    + \frac{\eps k \rank{S}}{2\tau} \leq %
    \frac{\epsmore k}{\tau} \rank{S}
  \end{align*}
  for all closed sets $S$. Thus $(\matroid, \rcap)$ can be
  fractionally covered by $\epsmore k / \tau$ bases. This proves \cref{RRC<}.

  For \cref{RRC>}, suppose $(\matroid, \capacity)$ cannot be
  fractionally covered by less than $k$ bases. Then there is a closed
  set $S$ such that
  \begin{math}
    \capacity{S} \geq k \rank{S}.
  \end{math}
  We have
  \begin{align*}
    \rcap{S} \geq \frac{1 - \eps/2}{\tau} \capacity{S} - \frac{\eps
    k}{2 \tau} \rank{S}
    \geq                        %
    \frac{\epsless k}{\tau} \capacity{S}.
  \end{align*}
  Thus $(\matroid, \rcap)$ cannot be fractionally covered by
  $\epsless k / \tau$ bases. This establishes \cref{RRC>} and
  completes the proof.
\end{proof}

\begin{remark}
  As mentioned above, for approximating the base packing problem
  specifically, \cite{Karger1998} already provides a lemma that allows
  us to reduce real-valued capacities to small integer capacities. The
  construction in \cite{Karger1998} is slightly different; in
  \cite{Karger1998}, each random capacity $\rcap{e}$ is sampled
  independently from a Poisson distribution of mean
  $p \cdot \capacity{e}$, for a parameter $p > 0$ with
  $p \geq c \eps^2 / k \ln{n}$ for a sufficiently large constant
  $c$. Overall the net effect is the same as the rounding-based
  construction for $\rcap$ that we analyze here. Despite the overlap
  with \cite{Karger1998} we include the proofs of \cref{RRP<,RRP>} as
  we find them interesting for the following reasons. First, the proof
  techniques here are unified with the proofs for the other matroid
  problems in \cref{RRU,RRC<,RRC>}. (Conversely, the proof techniques
  in \cite{Karger1998} did not seem as useful for these other
  problems.)  Second, the proofs here are different and arguably
  simpler than in \cite{Karger1998} as it does not depend on the
  random contraction algorithm.
\end{remark}

\subsection{Maximum capacity packings of independent sets and forests}

\begin{theorem}
  For real-valued capacities, a $\epsless$-approximation to the
  value of the maximum (fractional) $k$-fold matroid union, along
  with an $\epsmore$-approximately minimum dual solution, can be
  computed with high probability in
  \begin{math}
    \RandomizedApxMatroidUnionTime
  \end{math}
  randomized time. (Note that $\opt / k \leq \therank$.)
\end{theorem}

\begin{proof}
  We apply \cref{general-sparsification} to reduce the problem
  to integer capacities and optimum value
  $\parof{\opt / k} \log{n} / \eps^2$ with high probability. We then
  apply the $\epsless$-approximation algorithm for integer capacities
  from \cref{apx-matroid-union}.  The running time follows
  from \cref{apx-matroid-union}.
\end{proof}

The same reduction but for graphic matroids gives the following.

\begin{theorem}
  In an undirected graph with real-valued edge capacities, an
  $\epspm$-approximation to the maximum capacity $\opt$ that can be
  covered by a fractional packing of $k$ forests can be computed in
  \begin{math}
    \RandomizedApxForestUnionTime
  \end{math}
  randomized time. (Note that $\opt / k \leq n-1$.)
\end{theorem}

\subsection{Matroid base packing, matroid membership, and network
  strength}

\begin{theorem}
  \label{randomized-matroid-strength}
  For real-valued capacities, an $\epspm$-approximation to deciding if
  the matroid strength is (greater or less than) $k$ can be computed
  with high probability in randomized time bounded by
  \begin{math}
    \RandomizedApxBasePackingTime
  \end{math}
  independence queries.
\end{theorem}

\begin{proof}
  By either \cref{general-sparsification} or the techniques in
  \cite{Karger1998}, we can reduce the problem to packing $k / \tau$
  bases into integer capacities, with high probability. We then apply
  the $\epspm$-approximation algorithm for packing
  $k / \tau = \bigO{\ln{n} / \eps^2}$ bases with integer capacities
  given by \cref{apx-base-packing}. The running time follows from
  \cref{apx-base-packing}.
\end{proof}

The approximate algorithm for deciding matroid strength can be
extended to an approximation algorithm for approximating the matroid
strength via binary search. Here we present a modified algorithm that
carefully modifies the error parameters to reduce the standard
logarithmic overhead.

\begin{theorem}
  \label{search-matroid-strength}
  For real-valued capacities between $1$ and $U$, an
  $\epspm$-approximation to the matroid strength can be computed in
  time bounded by
  \begin{math}
    \RandomizedApxMatroidStrengthTime
  \end{math}
  independence queries.
\end{theorem}

\begin{proof}
  At the outset, we know that the matroid strength is between $1$ and
  $n U / \therank$.

  Now, for $i \in \nnintegers$, let $\eps_i = 2^{-i}$. A
  $\parof{1 + \eps_0}$-approximation to the strength can be obtained
  with high probability by combining a binary search of depth
  $\bigO{\log{n U / \therank}}$ with the approximate decision
  algorithm in \cref{randomized-matroid-strength} with error
  parameter a constant factor small than $\eps_0$. For $i \geq 1$,
  given a $\parof{1 + \eps_{i-1}}$-approximation for the strength, we
  can compute a $\parof{1 + \eps_i}$-approximation via a binary search
  of constant depth, with each probe making a call to
  \cref{randomized-matroid-strength} with error parameter
  $\bigOmega{\eps_i}$. Eventually we obtain a
  $\parof{1 + \eps_i}$-approximation where $\eps_i = 2^{-i}$ is at
  most the input error parameter $\eps$, as desired.

  We now bound the running time. The first set of $\log{n}$ calls to
  \cref{randomized-matroid-strength} with constant error parameter takes
  \begin{align*}
    \bigO{n \log{n U / \therank} + r \log{n} \log{\therank} \log{n U /
    \therank}}\text{-query}
    \labelthisequation{search-strength-1}
  \end{align*}
  time.  Thereafter we have a constant number of calls to
  \cref{randomized-matroid-strength} for each $\eps_i$ between $1$
  and $\eps/2$. For the leading $\bigO{n}$ term, all these calls add up
  to $\bigO{n \log{1/\eps}}$ work which is dominated by
  $\bigO{n \log{n}}$ above. For the second term of the form
  $\therank \ln{n} \ln{\therank / \eps} / \eps^3$, the sum over all
  $\eps_i$'s is dominated by the smallest $\eps_i$ which is
  $\bigOmega{\eps}$, given
  \begin{align*}
    \bigO{\therank \ln{n} \ln{\therank / \eps} / \eps^3}
    \labelthisequation{search-strength-2}
  \end{align*}
  work in total. Summing together
  \cref{equation:search-strength-1,equation:search-strength-2} gives
  the claimed running time.
\end{proof}

\subsection{Matroid base covering, matroid membership, and arboricity}

\labelsection{apx-matroid-membership}

Fractional base covering can be posed as a decision problem where,
given a capacitated matroid and an additional parameter $k$, the goal
is to decide if the capacities can be fractionally covered by $k$
bases.  The important special case of $k = 1$ is equivalent to testing
if a fractional point $x \in \nnreals^{\groundset}$ lies in the
independent set polytope. This problem is called matroid membership.

For deciding fractional base covering, we may assume that $k = 1$
without loss of generality. For a fixed error parameter
$\eps \in (0,1)$, a \emph{$\epspm$-approximation} to the matroid
membership problem is defined as a correct output that either (a) the
matroid can be covered by $1 + \eps$ bases, or (b) the matroid can be
covered by $1 - \eps$ bases. Note that either option is allowed when
the fractional covering number is between $1-\eps$ and $1 + \eps$. We
obtain the following running time for approximating matroid
membership.

\begin{theorem}
  \labeltheorem{apx-matroid-membership}\label{apx-matroid-membership}
  An $\epspm$-approximation to the matroid membership problem can be
  computed with high probability in
  \begin{math}
    \bigO{n + \therank \ln{n} \ln{\therank / \eps} / \eps^3}
  \end{math}
  randomized time.
\end{theorem}
\begin{proof}[Proof sketch]
  By \cref{general-sparsification}, we can reduce $\epspm$ matroid
  membership (with real capacities) to $\epsmore$-approximate integral
  base covering with $k = \bigO{\ln{n} / \eps^2}$ bases. The running
  time now follows from \cref{apx-base-covering}.
\end{proof}

For graphic matroids, recall that the fractional covering number is
called the arboricity. The following matches the theorem for matroid
membership above except for graphic matroids. The reduction is the
same except now we apply the corresponding algorithm for the graphic
matroid.

\begin{theorem}
  \label{apx-test-arboricity}
  An $\epspm$-approximation to deciding if a point is in the forest
  polytope can be computed with high probability in
  \begin{math}
    \RandomizedApxTreeCoveringTime
  \end{math}
  randomized time.
\end{theorem}

One may also want to find the maximum value $k$ for which a
capacitated matroid. The following uses \cref{apx-matroid-membership}
as a black box and is slightly better than one obtains by directly
plugging into a straightforward binary search.

\begin{theorem}
  \label{apx-search-matroid-arboricity}
  For real-valued capacities, an $\epspm$-approximation to the minimum
  $k$ by which a matroid can be covered by $k$ fractional bases  can
  be computed with high probability in running time bounded by
  \begin{math}
    \bigO{\parof{n + \therank \ln{\therank}} \ln{n} + \therank \ln{n}
      \ln{\therank / \eps} / \eps^3}
  \end{math}
  independence queries.
\end{theorem}

\begin{proof}
  At the outset, we know that the arboricity is between the maximum
  capacity of any element and the sum of capacities over all elements,
  which are within a factor $n$ of each other.

  Similar to \cref{search-matroid-strength}, for
  $i \in \nnintegers$, let $\eps_i = 2^{-i}$. A
  $\parof{1 + \eps_0}$-approximation to the strength can be obtained
  with high probability by combining a binary search of depth
  $\bigO{\log{n}}$ with the approximate decision algorithm in
  \reftheorem{apx-matroid-membership} with constant error
  parameter. For $i \geq 1$, given a
  $\parof{1 + \eps_{i-1}}$-approximation for the strength, we can
  compute a $\parof{1 + \eps_i}$-approximation with a binary search of
  constant depth. Each probe making a call to
  \reftheorem{apx-matroid-membership} with error parameter
  $\bigOmega{\eps_i}$. Eventually we obtain a
  $\parof{1 + \eps_i}$-approximation where $\eps_i = 2^{-i}$ is at
  most the input error parameter $\eps$, as desired.

  We now bound the running time. The first set of $\log{n}$ calls to
  \reftheorem{apx-matroid-membership} with constant error parameter
  take
  \begin{align*}
    \bigO{n \log{n} + r \log{n} \log{\therank}}
    \labelthisequation{search-arboricity-1}
  \end{align*}
  time.  Thereafter we have a constant number of calls to
  \reftheorem{apx-matroid-membership} for each $\eps_i$ between $1$
  and $\eps/2$. For the leading $\bigO{n}$ term, all these calls add
  up to $\bigO{n \log{1/\eps}}$ work which is dominated by
  $\bigO{n \log{n}}$ above. For the second term of the form
  $\therank \ln{n} \ln{\therank / \eps} / \eps^3$, the sum over all
  $\eps_i$'s is dominated by the smallest $\eps_i$ which is
  $\bigOmega{\eps}$, hence
  \begin{align*}
    \bigO{\therank \ln{n} \ln{\therank / \eps} / \eps^3}
    \labelthisequation{search-arboricity-2}
  \end{align*}
  work in total. Adding together
  \cref{equation:search-arboricity-1,equation:search-arboricity-2} gives
  the claimed running time.
\end{proof}

Applying the same modified binary search to the graphic matroid gives
the following randomized algorithm for estimating the arboricity of a
graph.
\begin{theorem}
  \label{apx-search-arboricity}
  An $\epspm$-approximation to the strength of a graph can be computed
  with high probability in
  \begin{math}
    \bigO{m \log{n} \ack{n} + n \log{n} \parof{\log{n} + \parof{\log
          \log{n} + \log{1/\eps}} \ack{n}} / \eps^3}
  \end{math}
  randomized time.
\end{theorem}


\section{Faster exact algorithms via augmentation}

\labelsection{augmenting-paths}

In this final section, we describe an augmenting path subroutine for
$k$-fold matroid union and use it obtain a faster exact algorithms
when $r \geq k^{1+o(1)}$.

\begin{lemma}
  \label{search-augmenting-path}
  \label{matroid-union-augmentation}
  A packing of $k$ independent sets produced by the
  $\epsless$-approximate $k$-fold matroid union data structure can be
  extended to an optimum solution in time bounded by
  $\bigO{\min{n + \opt \log{\therank},\therank k \log{k \therank}}}$
  queries per additional element.
\end{lemma}

Combining \cref{search-augmenting-path} with the
$\epsless$-approximation algorithm, for appropriate choice of $\eps$,
leads to the following running time which is faster in the regime
where $r \geq k^{1+o(1)}$. In particular, in the unweighted setting
where $\opt \leq kr \leq n$, we have a subquadratic upper bound of
$\apxO{n^{3/2}}$ independence queries.
\begin{theorem}
  \label{matroid-union-with-augmenting-paths}
  A maximum capacity packing of $k$ independent sets can be computed
  in time bounded by
  \begin{math}
    \MatroidUnionWithAugmentingPathsTime
  \end{math}
  independence queries, where
  $\MatroidUnionWithAugmentingPathsNotation$.
\end{theorem}

\begin{proof}
  Let $\eps > 0$ be a parameter to be determined.  A
  $(1-\eps)\opt$-capacity packing can be computed in
  \begin{math}
    \ApxMatroidUnionTime
  \end{math}
  independence queries. This can be augmented to an optimal solution
  in $\bigO{n' + \opt \log{\therank}}$ time per augmentation. Thus the
  total running time is
  \begin{align*}
    \bigO{n + \opt \log{kr} / \eps + \eps \opt \parof{n' + \opt \log{\therank}}}.
  \end{align*}
  The last two terms are balanced by taking
  $\eps = \sqrt{\parof{n' + \opt \log{\therank}} / \log{kr}}$ gives
  the claimed running time.  Here we note that a constant factor
  approximation for $\opt$ can be obtained by running the
  approximation matroid union algorithm with constant $\eps$, and this
  suffices to balance the terms up to constant factors.
\end{proof}

It remains to prove \cref{matroid-union-augmentation}.

\subsection{Initialization from matroid push-relabel}

We need to initialize our algorithm with a packing of $k$ independent
sets $\Is$, whereas the $k$-fold matroid union algorithm from
\refsection{impr} directly produces $k$ bases $\Bs$.  Such a packing
$I_1,\dots,I_k$ with the same total capacity can be easily obtained
from $\Bs$ by dropping overpacked elements from some of the bases
until there are no overpacked elements. However to preserve certain
useful structures of $\Bs$ we carefully remove overpacked elements
from the bases in the following greedy fashion.

Initially, we set $I_1 = \basei[1], \dots, I_k = \basei[k]$. While
there is an overpacked element $e$ (\wrt $\Is$), we remove $e$ from
then independent set $I_j$ of maximum index $j$. It is easy to see
that at termination there are no overpacked elements while the
objective value is preserved. By removing overpacked elements in such
a fashion we also gain the following critical properties which we now
define.

\begin{definition}
  \labeldefinition{packing-invariants}
  Let $\Is$ be a packing of $k$ independent sets. We say that $\Is$ is
  \emph{maximal} if for all uncovered elements $e$ and all independent
  sets $\Ii$ we have $e \in \Ii$. We say that $\Is$ are in
  \emph{decreasing order} if their spans are; that is,
  $\indi[i+1] \subseteq \spn{\indi}$ for $i = 1,\dots,k-1$.
\end{definition}

\begin{lemma}
  Let $\Is$ be a packing of $k$ independent set obtained from the
  bases $\Bs$ output by the $k$-fold matroid union push-relabel
  algorithm in the greedy fashion described above. Then:
  \begin{mathproperties}
  \item \label{maximal-packing} $\Is$ is a maximal packing.
  \item \label{decreasing-packing} $\Is$ are in decreasing order.
  \end{mathproperties}
\end{lemma}

\begin{proof}
  Fix the configuration of $\Bs$ and
  $\level: \groundset \to \nnintegers$ at the end of the matroid
  push-relabel algorithm.

  We first claim that any overpacked element $e$ has $\level{e} =
  0$. Indeed, initially all elements have $\level{e} = 0$, and
  elements can only be relabeled when they are uncovered. Meanwhile,
  an element can only be overpacked by the initial configuration, so
  an element that is overpacked at termination was overpacked -- and
  never uncovered -- all along.

  Now, each $I_i$ will contain all elements from $B_i$ that are not
  overpacked. Since overpacked elements have level $0$, we
  have
  \begin{math}
    \basei_{\geq 1} \subseteq I_i \subseteq \basei
  \end{math}
  for all $i$. Meanwhile, any uncovered element has level $> 1$. By
  \cref{imp-span}, for all uncovered elements $e$ we have
  \begin{align*}
    e \in \level_{> 1} \subseteq \spn{\basei_{\geq 1}} \subseteq \spn{I_i}
  \end{align*}
  This establishes \cref{maximal-packing}.

  Next we show \cref{decreasing-packing}, which claims that the
  independent sets are in decreasing order.  For each independent set
  $I_i$, we can express $\Ii$ as the disjoint union of
  $\basei{\geq 1}$ and $\Ii{0} \defeq \Ii \cap \basei{0}$.  For each
  index $i \in [k-1]$, we have
  \begin{math}
    \basei[i+1]{\geq 1} \subseteq \spn{\basei{\geq 1}}
  \end{math}
  by the decreasing order of bases (see defintion
  \ref{decreasing-bases}) and
  \begin{math}
    \Ii[i+1]{0} \subseteq \Ii{0}
  \end{math}
  because overpacked elements are removed from the independent sets of
  maximum index.  Thus
  \begin{align*}
    \Ii[i+1] = \basei[i+1]{\geq 1} \cup \Ii[i+1]{0}
    \subseteq
    \spn{\basei{\geq 1}} \cup \Ii{0}
    \subseteq \spn{\Ii},
  \end{align*}
  as desired.
\end{proof}

\subsection{Greedy sparsification}

Before proceeding to describe the augmenting path algorithm, we point
out that by techniques by \cite{Karger1998}, one can assume
$n \leq \bigO{kr \ln{r}}$.

\begin{lemma}
  Let $I_1,\dots,I_{\ell}$ be a maximal packing of
  $\ell \geq k (1 + \ln{r})$ sets. Then the size of the maximum
  $k$-fold in the (smaller) capacitated matroid
  induced by $I_1,\dots,I_{\ell}$ is the same as in the input matroid.
\end{lemma}

\begin{proof}[Proof sketch]
  The proof is essentially the same as \cite{Karger1998} which
  focused on base packing instead. In the proof, one replaces the role
  of dual characterization for base packing \cite{Edmonds1965a} with
  the dual characterization for matroid union \cite{NashWilliams1967}.
\end{proof}

\cite{Karger1998} described how to construct such a packing greedily
in $\bigO{n + k \therank \ln{\therank}}$ independence queries. The
same construction extends here except we start with the maximal
packing $\Is$ given by the push-relabel algorithm and then extend it
greedily. Thus we have the following.

\begin{lemma}
  With a running time overhead of $\bigO{\min{n, k r \ln{r}}}$
  independence queries, we may assume that $n \leq \bigO{k r \ln{r}}$.
\end{lemma}

\subsection{Setting up the auxiliary graph}

\NewDocumentCommand{\ei}{G{e} e{_} G{i} e{^}}{%
  #1^{\smash{\parof{\IfNoValueTF{#4}{#3}{#4}}}}%
  \IfNoValueF{#2}{_{#2}}%
}%

\newcommand{\elast}{\ei^{\numM}}
\newcommand{\dlast}{\ei{d}^{\numM}}
\newcommand{\ej}{\ei^{j}}%
\newcommand{\di}{\ei{d}}%
\renewcommand{\dj}{\ei{d}^{j}}%

At a high-level, we have a packing of independent sets $\Is$, and the
immediate goal is to increase their total size. However it is not as
simple as finding an uncovered element $e$ to add to a set $I_k$ ---
all the uncovered elements are spanned by all the independent sets
because $\Is$ is maximal. To extend the total size of $\Is$ we may
have to shuffle many of the elements from among the $\Is$ to make room
for one more element.

As is well-known, such an augmentation can be found be searching for a
path in a directed auxiliary graph where each arc encodes a local
exchange such as replacing one element with another in an independent
set $\Ii$, or moving an element from one $\Ii$ to another. Here we
describe the auxiliary graph which is standard
\cite{Knuth1973,GreeneMagnanti1975}.\footnote{This auxiliary graph can
  also be interpreted through the lens of matroid intersection with
  some modifications.} We have two auxiliary vertices $s$ and $t$
which will act as the beginning and end of our search. For each
element $e$, we have $k$ auxiliary vertices $\ei^1,\dots,\ei^k$. Each
$\ei$ represents $e$ in relation to $\basei$ as will be made clear
by the arcs which we now describe.  We have four types of arcs.
\begin{enumerate}
\item \label{source-arc} $(s, \ei)$ where $e$ is uncovered and $\ei \notin B_i$. This arc
  represents trying to exchange $\ei$ into $B_i$.
\item \label{exchange-arc} $(\ei,\di)$ where $d \in B_i$, $e \notin B_i$, and
  $B_i - d + e$ is independent. This arc represents
  exchanging $e$ for $d$ in $B_i$.
\item \label{shuffle-arc} $(\di,\dj)$ where $d \in B_i$,
  $d \in \spn{B_j} \setminus B_j$, and $d$ is covered. This arc
  represents removing $d$ from $B_i$ and using the freed up capacity
  to initiate an exchange for $d$ into $B_j$.
\item \label{sink-arc} $(\ei, t)$ where $e$ is not spanned by
  $\basei$. This arc represents inserting $e$ into $\basei$.
\end{enumerate}
In total, the auxiliary graph has $\bigO{n k}$ vertices, and at most
$\bigO{n rk}$ arcs.

Paths from $s$ to $t$ in this graph have a very specific
graph. Between $s$ and $t$, the path consists of auxiliary vertices of
the form $\ei$ that alternate between those where $e \notin \basei$
and where $e \in \basei$. There are always an odd number of such
internal vertices, with one more of the former type.  Arc-wise, the
first arc is of type \ref{source-arc}, the last arc is of type
\ref{sink-arc}, and in-between the arcs alternate between type
\ref{exchange-arc} and type \ref{shuffle-arc}, starting with type
\ref{exchange-arc} and ending with type \ref{exchange-arc}.
All put together, an $(s,t)$-path in the auxiliary graph above
corresponds to a sequence of exchanges, plus one final insertion,
where the net effect is to increase the total size of $\Is$ by $1$.
We call such a path an \emph{augmenting path} if it maintains
feasibility. That is, the sets $\Is$ remain independent, and we do not
overpack elements.

$(s,t)$-paths are not necessarily augmenting paths.  On one hand, it
is easy to see that these operations will not violate any capacity
constraints, since any time a covered element is inserted into a set,
it is preceded by removing the same element from another set. As per
the feasibility of the $\Is$, while each exchange (encoded by an arc
of type \ref{exchange-arc}) would individually maintain the
independence of each set, all the exchanges taken together do not
necessarily maintain independence. However, the following criteria
outlines conditions in which a sequence of exchanges to the same
independent set does maintain independence.  This criteria is standard
and have long been used to justify matroid intersection and partition
algorithms. (Here the wording is from \cite{Cunningham1984}).

\begin{fact}
  \label{augmenting-path}
  Given a matroid $\defmatroid$, let $I \in \independents$, and let
  $e_1, \dots, \notin I$, $d_1,\dots,d_p \in I$, and optionally
  $e_{p+1} \notin I$ be distinct elements such that:
  \begin{mathproperties}
  \item For $i =0, \dots, p$, $I - d_i + e_i \in \independents$.
  \item \label{chordal-pair} For $0 \leq j < i \leq p$,
    $I - d_i + e_j \notin \independents$.
  \item $I + e_{p+1} \in \independents$ (when including $e_{p+1}$).
  \end{mathproperties}
  Then
  $I' = I - d_1 - \cdots - d_p + e_1 + \cdots + e_{p+1} \in
  \independents$.
\end{fact}

The exchanges corresponding to an $(s,t)$-path in the auxiliary graph
can only violate \cref{chordal-pair} out of the three properties in
\cref{augmenting-path}. Below, we design a subroutine that, given an
$(s,t)$-path in the auxiliary graph, extracts a subpath that is an
augmenting path, by efficiently identifying and removing violations to
\cref{chordal-pair}, as follows.

\begin{lemma}
  \label{extract-path-time}\label{prune-path-time}
  Given an $(s,t)$-path of length $\ell$ in the auxiliary graph, one
  can compute an augmenting subpath in running time bounded by
  $\bigO{\ell \log{\therank}}$ independence queries.
\end{lemma}

\begin{proof}
  Fix an independent set $I_i$. Suppose the path encodes inserting
  $e_1,\dots,e_p \notin I_i$ in exchange for $d_1,\dots,d_p \in I_i$,
  respectively and in sequence. This means that the auxiliary arcs
  $(\ei_j,\di_j)$ for $j = 1,\dots,p$ appear in the auxiliary path in
  that order.  If
  $I_i - d_1 - \cdots - d_p + e_1 + \cdots + e_p \notin
  \independents$, then by fact \ref{augmenting-path}, there must be
  indices $1 \leq j_1 < j_2 \leq p$ such that
  $I_i - d_{j_2} + e_{j_1} \in \independents$. This exchange implies
  $\parof{\ei_{j_1}, \di_{j_2}}$ is an arc in the auxiliary graph, and
  so we can shorten our $(s,t)$-path by replacing all arcs between
  (and including) $(\ei_{j_1}, \di_{j_1})$ and $(\ei_{j_2},\di_{j_2})$
  with $\parof{\ei_{j_1}, \di_{j_2}}$. Let us call such a pair
  $(e_{j_1}, d_{j_2})$ a \emph{chordal pair} as it represents a chord
  \wrt our $(s,t)$-path.  Our goal is to repeatedly identify chordal
  pairs $(j_1,j_2)$ and shorten the path until we arrive at a shorter
  sequence of exchange for $I_i$ that is feasible.

  A pair of indices $(j_1,j_2)$ as described above can be identified
  efficiently as follows.  For $j_1 = p$ down to $1$, we binary search
  for the first index $j_2$ such that
  $I - d_1 \cdots - d_{j_1} - d_{j_2} + e_{j_1} \in
  \independents$. Then $j_2$ is the first index such that $d_{j_2}$ is
  in the circuit of $I + e_{j_1}$. Now, if $j_1 = j_2$, then this
  verifies that $I - d_{j} + e_{j_2} \notin \independents$ for all
  $j < j_2$, as desired. Otherwise $(j_1,j_2)$ represents a chord
  which can be used to shorten the $(s,t)$-path. We shortcut the path
  at $(j_1,j_2)$. We then decrease $j_1$ and continue the search,
  short cutting or validating each $e_{j_1}$ until we have certified
  that there are no chordal pairs remaining. The total running time to
  prune all chordal pairs for $I_i$ is bounded above by
  $\bigO{p \log{\therank}}$, where $p$ refers to the number of
  exchanges originally encoded in the path, before pruning.

  The description above was for the case where the path encoded a
  sequence of exchanges and not an additional insertion. However the
  discussion extends immediately to the case where the path encodes a
  sequence of $p$ exchanges plus an additional insertion. Again we can
  prune all chordal pairs for a fixed independent set $I_i$, and guarantee
  that the remaining sequence of exchanges and insertion maintains
  the independence of $I_i$, in running time bounded above by $\bigO{p
    \log{\therank}}$ independence queries.

  Overall the algorithm process each independent set $I_i$ one at a
  time. For each $I_i$ we prune chordal pairs and shorten the
  $(s,t)$-path so that the remaining exchanges maintain the
  independence of $I_i$. The total running time over all sets $I_i$ is
  bounded above $\bigO{\ell \log{\therank}}$ independence queries,
  since $\bigO{\ell}$ counts the total number of exchanges in the
  original path over all independent sets $I_i$.
\end{proof}

\subsection{Efficiently searching for $(s,t)$-paths in the auxiliary
  graph}

We have now set up an auxiliary graph and shown that any $(s,t)$-path
can be efficiently converted to an augmenting path. It remains to
design an efficient method to find an $(s,t)$-path. Of course one
could construct the graph explicitly by testing for the presence for
every arc above, but as remarked above the size of the overall graph
is bigger than the desired running time.

Recall that our goal is to find an $(s,t)$-path or conclude that no
such paths exist. In particular we are not strictly required to test
and traverse all the arcs in the graph. Our goal is to develop a
search our algorithm with running time bounded by a number of
independence queries proportional (up to logarithmic terms) to the
total number of vertices in the graph.

We will show how to modify BFS to this purpose.\footnote{Other marking
  based search algorithms such as DFS could have been used instead.}
Recall that BFS marks vertices when they are first visited before
adding them to the queue of elements to be searched next. Of course,
the marks record that a vertex has been visited and prevents infinite
loops. In the context of our implicit auxiliary graph we can also
avoid testing for the existence of an arc if the head of the arc is
already marked.

For the sake of efficiency we will impose the following invariant on
the set of marked vertices.  For each independent set $I_i$ in our
packing, let $\markedi$ denote the set of auxiliary vertices $\ei$,
where $e \in \Ii$, that have been marked. We will strictly adhere to
the invariant that the marked sets $\Ms$ are in decreasing order (in
the same sense as $\Is$, cf.\ \refdefinition{packing-invariants}).  To
maintain this invariant we introduce the following subroutine.

\paragraph{Predecessor search.}
Whenever we are about to mark an element $\ei$ where $e \in B_i$, we
need to ensure that $e \in \spn{M_j}$ for all $j < i$. In this case we
launch a \emph{predecessor search}, or \emph{pre-search} for short, at
$e_i$, which executes the following steps.

\begin{quote}
  \emph{Pre-searching an auxiliary element $\ei$:} Let $e \in B_i$.
  \begin{steps}
  \item While $e \notin \spn{\Mi[i-1]}$:
    \begin{steps}
    \item For each element $d \in \Ii[i-1] \setminus \spn{M_j}$ such that
      $\Ii[i-1] - d + e \in \independents$.
    \item Recursively pre-search $d_j$.
    \item Mark $d_j$, and record the arc $(e_i,d_j)$, and add $d_j$ to
      the (outer, BFS) search queue.
    \end{steps}
  \end{steps}
\end{quote}

The steps above are described at a high-level and we will discuss
concrete details and issues of efficiency later. First we establish
why the (high-level) steps above maintain $\Ms$ in decreasing
order. (The reason we order this fact before analyzing concrete
implementation details is because we will use the decreasing order of
the $\Ms$ to make the implementation more efficient.)

\begin{lemma}
  Suppose $\Is$ and $\Ms$ are both in decreasing order, and let
  $e \in B_i$. Then pre-searching $\ei$ maintains the $\Ms$ in
  decreasing order. If $i > 1$, we also have
  $\Mi + e \subseteq \spn{\Mi[i-1]}$.
\end{lemma}

\begin{proof}
  We prove the claim by induction on $i$. If $i = 1$ then there is
  nothing to do and the claim is vacuous.  Now, let $i > 1$ and assume
  the claim holds for $i-1$. Consider a pre-search to an auxiliary
  vertex $\ei$ where $e \in I_i$. Recall that $I_{i-1}$ spans $e$
  because $\Is$ is in decreasing order. Consequently if $\Mi[i-1]$
  does not span $e$ then there are still other elements $d$ in the
  circuit of $\Ii[i-1] + e$ that are not in $\Mi[i-1]$. As long as
  $\Mi[i-1]$ does not span $e$, the pre-search subroutine repeated
  selects such an element $d$ and calls pre-search on $\di^{i-1}$
  before marking $\di^{i-1}$. By induction the pre-search on
  $\di^{i-1}$ ensures that we can safely mark $\di^{i-1}$ while
  preserving $\Ms$ in decreasing order. The pre-search at $e$ ends
  with $\Ms$ still in decreasing order, and with $\Mi[i-1]$ spanning
  $e$.
\end{proof}

\paragraph{The implicit search algorithm.}
We now describe the search algorithm which employs pre-search as a
black box. Here the challenge is to search the auxiliary graph without
knowing all the arcs explicitly. At a high-level, by leveraging the
decreasing order of both $\Is$ and $\Ms$, up to logarithmic factors,
we are able to limit our queries to those that produce arcs to as yet
univisited vertices.

The search algorithm behaves differently depending on the type of
auxiliary vertex we're search. Here we have three types. The first is
at the source $s$.  The second is at an at auxiliary vertices $\ei$
where either $e \in \Ii$ or $i = k$. In particular we do not call
search directly on auxiliary vertices $\ei$ where $e \notin \Ii$ and
$i < k$, for efficiency reasons made clearer below.

We start with the search routine for $s$, which simply searches every
uncovered element $e$.

\begin{quote}
  \emph{Searching at the source $s$:}
  \begin{enumerate}
  \item For each uncovered element $e$, if $\elast$ is unmarked, then
    mark $\elast$, record the arc $(s, \elast)$, and call search on
    $\elast$.
  \end{enumerate}
\end{quote}

Next we describe the search algorithm for elements of the form $\ei$.

\begin{quote}
  \emph{Searching at an auxiliary vertex $\ei$:}
  \begin{enumerate}
  \item If $e \notin \spn{\lastI}$:
    \begin{enumerate}
    \item Let $j$ be the first index such that $e \notin \spn{\indj}$.
    \item Mark $\ej$, mark $t$, and record the arcs $(\ei, \ej)$ and
      $(\ej, t)$. Signal that we have found an $(s,t)$-path.
    \end{enumerate}
  \item Otherwise, while $e \notin \spn{\lastM}$:
    \begin{enumerate}
    \item Let $d \in \lastI \setminus \lastM$ be such that
      $I - d + e \in \independents$.
    \item Pre-search $\dlast$.
    \item Mark $\elast$ and $\dlast$. Record the arc $(\ei, \elast)$
      and $(\elast, \dlast)$. Add $\dlast$ to the search
      queue.
    \end{enumerate}
  \end{enumerate}
\end{quote}

\begin{lemma}
  After searching an auxiliary vertex $\ei$, where $e \in \indi$, we
  have the following:
  \begin{enumerate}
  \item If there is an index $j$ such that $e \notin \spn{I_j}$, then
    for the first such index $j$, the path $(\ei, \ej, t)$ is recorded
    and the vertices $\ej$ and $t$ are marked.
  \item If not, then we have $e \in \spn{\Mi[j]}$ for all indices
    $j$.
  \end{enumerate}
\end{lemma}
\begin{proof}
  The first case is straightforward from the code. In the second case,
  we have $e \in \spn{\Ii}$ for all $i$. The search at $e$ exits only
  when $e \in \spn{\lastM}$. Because $\Ms$ are in decreasing order,
  we then have $e \in \spn{\markedi}$ for all $i$.
\end{proof}

\begin{lemma}
  After searching $\elast$ for an uncovered element $e$, we have
  $e \in \spn{\Mi}$ for all $i$.
\end{lemma}
\begin{proof}
  We know that $e in \spn{\lastI}$ because $e$ is uncovered and $\Is$
  is a maximal packing.  The search exits only when
  $e \in \spn{\lastM}$. Because $\Ms$ are in decreasing order, we then
  have
  $e \in \spn{\markedi}$ for all $i$.
\end{proof}

\begin{lemma}
  The search finds a path from $s$ to $t$ if one exists.
\end{lemma}

\begin{proof}
  The algorithm records arcs that jointly contain paths to all marked
  vertices. Thus the algorithm marks $t$ and signals that an
  $(s,t)$-paths is found, then the record arcs contain the desired
  arc. (Parent pointers, or a graph search through the recorded arcs,
  will produce an $(s,t)$-path.)

  Observe that if the search does not signal an $(s,t)$-path, it will
  still mark all reachable auxiliary vertices except possibly for an
  auxiliary vertex $\ei$ where $e$ is uncovered and $e \notin
  \Ii$. However the preceding lemma ensures that $e \in \spn{\Mi}$, so
  any of the auxiliary vertices reachable from $\ei$ have still been
  explored. Thus if $t$ is reachable, then it will be marked and an
  $(s,t)$-path will be found, as desired.
\end{proof}

\paragraph{Efficiency of search and pre-search:}

\begin{lemma}
  \label{search-time}
  A search from $s$  takes running time bounded
  by
  \begin{math}
    \bigO{n + \opt \log{\therank} + \log{k}}
  \end{math}
  independence queries.
\end{lemma}
\begin{proof}
  Let $\ell$ denote the total number of auxiliary vertices of the
  $\ei$ that are marked. Note that $\ell \leq \bigO{n + \opt}$ because
  if we mark an auxiliary vertex $\ei$ with $e \notin I_i$ and
  $i < k$, then there is also an outgoing arc to some $\dj$ were
  $d \in I_j$.  Each search or pre-search routine at an element $e$
  consists of a constant number of queries plus an unspecified number
  of queries per auxiliary vertex that gets marked, in order to search
  for the element that is marked.  The constant number of queries per
  search add up to at most $\bigO{n + \ell}$ in total, because each
  auxiliary vertex is pre-searched or searched at most once (before or
  after it is marked). Next we explain how to implement the latter
  category of queries in a logarithmic number of queries per marked
  element. Here we have two types of searches for the next marked
  element.

  The first is to identify elements $d \in \Ii \setminus \markedi$
  such that $I - d + e \in \independents$, given $e$ and an index
  $i$. This can be done by maintaining $\Ii$ in any order such that
  $\markedi$ comes first, and searching for the first prefix of
  this ordering that spans $e$. The last element in this ordering
  gives an element $d$. The search requires $\bigO{\log{\therank}}$
  oracle queries. (Note that it is easy to maintain this ordering as
  elements are added to $\markedi$.)

  The second type is to identify an index $j$, as small as possible,
  such that $e \notin \spn{I_j}$. Since $\Is$ is in decreasing order,
  we can binary search for the first index $j$ with $\bigO{\log k}$
  probes. Each probe corresponds to $1$ independence
  query. Additionally, this search also ends the search, so we only
  search for such an index $j$ once.
\end{proof}

\paragraph{Preserving $\Is$ in decreasing order.} Next we address the
fact that our particular choice of augmentations keeps $\Is$ in
decreasing order.
\begin{lemma}
  An augmentation induced by the search algorithm keeps $\Is$ in
  decreasing order.
\end{lemma}
\begin{proof}
  Suppose the end of the path encodes inserting an element $e$ in
  $\Ii$. We first note that the preceding exchanges do not effect the
  span. Then, when inserting $e$ into $\Ii$, we know that
  $e \in \spn{I_{i-1}}$ for all $j < i$ by choice of $i$. Thus
  $I_i + e \in \spn{I_{i-1}}$ and we preserve the decreasing order.
\end{proof}

\paragraph{Putting it all together.}

Together, \cref{prune-path-time,search-time} gives the overall running
time to find an augmenting path. (To apply \cref{prune-path-time}, we
note that the length of any $(s,t)$-path is at most $\bigO{\opt}$.)
This completes the proof of \cref{search-augmenting-path}.

\subsection{Graphic matroids}

We conclude the section by translating the oracle-based running time
into concrete running times for the graphic matroid. The ideas here
are similar to those in \refsection{spanning-trees}. By maintain
disjoint union data structures over $I_i$, we can maintain and
implement independence queries for $I_i$ with $\ack{n}$ overhead. By
also managing each $I_i$ in a link-cut trees, we can retrieve unmarked
edges $d \in I_k \setminus M_k$ for exchanging in $\bigO{\log n}$ time
(bypassing the binary search from the oracle model). Retracing the
proofs of \cref{prune-path-time,search-time}, and keeping in mind that
the maximum length of a path is $\opt$, shows that it takes
$\bigO{m' \ack{n} + \opt \log{nk}}$ time per augmenting path for
$m' = \min{m, n \log{n k}}$. Balancing the choice of $\eps$ with the
$\ApxForestUnionTime$ running time for the $\epsless$-approximation, we
obtain the following.

\begin{theorem}
  \labeltheorem{forest-union-b}
  For graphs with integer edge capacities, a forest
  packing of maximum total size can be computed in
  \begin{align*}
    \bigO{m \ack{n}
    +
    \sqrt{\opt \parof{\min{m, n \log{n k}} \ack{n} + \opt \log{n k}}
    \log{n} \parof{\log{n} + \log{k} \ack{n}}}
    }
  \end{align*}
  time.
\end{theorem}

A simplified running time for connected graphs, observing that
$\opt \geq \bigOmega{n}$ for $k \geq 1$, is
\begin{align*}
  \ForestUnionTimeB.
\end{align*}


\printbibliography

\appendix

\section{Additional background}

\label{additional-background}

In this section, we provide some additional background omitted from
\cref{results} and not appearing elsewhere in the paper.

Recall the problem of maximizing the size of a forest packing (i.e.,
$k$-fold matroid union problem in the graphic matroid). There are
several previous works on this problem
\cite{Imai1983a,RoskindTarjan1985,GabowStallmann1985,GabowWestermann1992}.
\cite{GabowWestermann1992} obtained several running times that are
each optimal for some range of parameters, of which we highlighted the
most comparable in \cref{results}.

Packing spanning trees and network strength have important additional
connections via the classical theorem of
\citet{Tutte1961a,NashWilliams1961b} which we now describe.
Given a partition of $(V_1,\dots,V_k)$ of vertex set $V$, the
\defterm{strength} of the partition is defined as the ratio
\begin{align*}
  \sizeof{\cut{V_1,\dots,V_k}} / \parof{k-1},
\end{align*}
where $\cut{V_1,\dots,V_k}$ denotes the set of edges cut by the
partition $V_1,\dots,V_k$.  One interpretation of the strength of a
partition, given by \citet{Cunningham1985c} in the context of network
vulnerability, is that the strength reflects a cost per
additional connected component created by the cut.  The
\defterm{(network) strength} of a graph is defined as the minimum
strength over all partitions. Network strength was proposed as a
measure of network vulnerability by
\citet{Gusfield1983}. The definitions extend naturally to
  positive edge capacities.  \cite{Tutte1961a,NashWilliams1961b}'s
theorem shows that the network strength equals the maximum size of any
(fractional) tree packing; the maximum size of any integral tree
packing is the floor of the network strength.

There are several works on packing spanning trees in capacitated and
uncapacitated graphs
\cite{Imai1983a,RoskindTarjan1985,GabowWestermann1992,Gabow1991a,Cunningham1984,Trubin1991,Barahona1995,GabowManu1998}.
Besides the running times \cref{results}, \cite{GabowManu1998} obtains
a strongly polynomial running time of $\bigO{n^3 m \log{n^2 / m}}$.
Beside via packing spanning trees, there is further work on computing
the strength directly
\cite{Cunningham1985c,Gusfield1991,Barahona1992,ChengCunningham1994,Gabow1998}
via flow or parametric flow, culminating in
$\bigO{n^2 m \log{n^2/m}}$-time algorithms for computing the strength
via parametric flow (and related techniques)
\cite{ChengCunningham1994,Gabow1998}.  For covering by trees, besides
the running times mentioned in \cref{results}, \cite{GabowManu1998}
also obtains an $\bigO{n^3 m \log{n^2 / m}}$ running time for covering
by trees in the capacitated graphs, the same as for packing.


\end{document}